\documentclass[11pt]{article}
\usepackage{todonotes}
\usepackage{amsthm}
\usepackage{amsmath}
\usepackage{amsfonts}
\usepackage{amssymb}
\usepackage{algorithm2e}
\usepackage{xcolor}
\usepackage{lipsum}
\usepackage[margin=1in]{geometry}
\usepackage[numbers,sort]{natbib}
\usepackage{hyperref}
\usepackage{mdframed}
\usepackage{tikz}
\usetikzlibrary{decorations.pathreplacing}
\usetikzlibrary{decorations.pathmorphing}
\tikzset{snake it/.style={decorate, decoration=snake}}
 
 \usepackage{float}
 
 \newtheorem{definition}{Definition}
 \newtheorem{lemma}{Lemma}
 \newtheorem{proposition}[]{Proposition}
 \newtheorem{theorem}[]{Theorem}

 \newcommand{\dist}{\texttt{dist}}
 \newcommand{\eps}{\varepsilon}
 \newcommand{\Oish}{\widetilde{O}}
 \newcommand{\zz}{\mathbb{Z}}
 \newcommand{\rr}{\mathbb{R}}
 \newcommand{\poly}{\text{poly}}
 \newcommand{\coord}{\normalfont\texttt{coord}}
 \newcommand{\proj}{\normalfont\texttt{proj}}

\title{New Additive Spanner Lower Bounds by an Unlayered Obstacle Product\footnote{This work was supported by NSF:AF 2153680.}}
\date{}

\author{
\begin{tabular}{c c}
   Greg Bodwin  & Gary Hoppenworth  \\
   University of Michigan & University of Michigan
   \\
     \texttt{bodwin@umich.edu}  &  \texttt{garytho@umich.edu} \\
\end{tabular}
}

\begin{document}
\maketitle
\thispagestyle{empty}
\begin{abstract}
For an input graph $G$, an \emph{additive spanner} is a sparse subgraph $H$ whose shortest paths match those of $G$ up to small additive error.
We prove two new lower bounds in the area of additive spanners:
\begin{itemize}
\item We construct $n$-node graphs $G$ for which any spanner on $O(n)$ edges must increase a pairwise distance by $+\Omega(n^{1/7})$.
This improves on a recent lower bound of $+\Omega(n^{1/10.5})$ by Lu, Wein, Vassilevska Williams, and Xu [SODA '22]. 

\item A classic result by Coppersmith and Elkin [SODA '05] proves that for any $n$-node graph $G$ and set of $p = O(n^{1/2})$ demand pairs, one can \emph{exactly} preserve all pairwise distances among demand pairs using a spanner on $O(n)$ edges. 
They also provided a lower bound construction, establishing that  this range $p = O(n^{1/2})$ cannot be improved.
We strengthen this lower bound by proving that, for any constant $k$, this range of $p$ is still unimprovable \emph{even if the spanner is allowed $+k$ additive error among the demand pairs}.
This negatively resolves an open question asked by Coppersmith and Elkin [SODA '05] and again by Cygan, Grandoni, and Kavitha [STACS 
'13] and Abboud and Bodwin [SODA '16].
\end{itemize}
At a technical level, our lower bounds are obtained by an improvement to the entire \emph{obstacle product} framework used to compose ``inner'' and ``outer'' graphs into lower bound instances.
In particular, we develop a new strategy for analysis that allows certain \emph{non-layered} graphs to be used in the product, and we use this freedom to design better inner and outer graphs that lead to our new lower bounds.
%
\end{abstract}


\pagebreak
\setcounter{page}{1}
\section{Introduction}

A basic question arising in robotics \cite{Cai94}, circuit design \cite{CKRS90, CKRS92, SCRS01}, distributed algorithms \cite{Awerbuch85, BCLR86}, and many other areas of computer science (see survey \cite{ABSHJKS20}) is to compress a graph metric $G$ into small space while approximately preserving its shortest path distances.
When this compression is achieved by a sparse subgraph $H \subseteq G$ whose distance metric is similar to that of $G$, we call $H$ a \emph{spanner} of $G$.
The setting of spanners on a linear or near-linear number of edges is considered particularly important in applications; that is, for an input graph on $n$ nodes, we often want spanners on $O(n)$ or perhaps $O(n^{1+\eps})$ edges \cite{ABSHJKS20}.

Spanners were first abstracted by Peleg and Upfal \cite{PU89jacm} and Peleg and Ullman \cite{PU89sicomp} after arising implicitly in prior work.
Their initial work studied spanners in the setting of \emph{multiplicative} error:
\begin{definition} [Multiplicative Spanners]
For a graph $G$, a subgraph $H \subseteq G$ over the same vertex set is a $\cdot k$ multiplicative spanner of $G$ if, for all nodes $s, t$, we have $\dist_H(s, t) \le k \cdot \dist_G(s, t)$.
\end{definition}

Optimal bounds for multiplicative spanners were quickly closed in a classic paper by Alth{\" o}fer, Das, Dobkin, Joseph, and Soares \cite{ADDJS93}.
For the specific case of $O(n)$-size spanners, they proved:
\begin{theorem} [\cite{ADDJS93}] \label{thm:multspan}
Every $n$-node graph has a $\cdot O(\log n)$ multiplicative spanner on $O(n)$ edges.
Moreover, there are graphs that do not have an $\cdot o(\log n)$ multiplicative spanner on $O(n)$ edges.
\end{theorem}
Thus, the question of compression by multiplicative spanners was closed.
However, for many problems, the paradigm of multiplicative error is considered unacceptable.
For example, $\cdot O(\log n)$ multiplicative blowup in travel time would be unacceptable for a cross-country trucking route.
This generated interest in the community in new error paradigms that perform better on the \emph{long} distances in the input graph.
Several new types of spanners were suggested and studied in the following years \cite{EP04, TZ06}.
The most optimistic was \emph{purely additive spanners}, where pairwise distances in the spanner increase only by an \emph{additive} error term:
\begin{definition} [Additive Spanners \cite{LS91}]
For a graph $G$, a subgraph $H \subseteq G$ over the same vertex set is a $+k$ additive spanner of $G$ if, for all nodes $s, t$, we have $\dist_H(s, t) \le \dist_G(s, t) + k$.
\end{definition}

Unfortunately, high-quality constructions of near-linear-size spanners with purely additive error remained elusive.
That is, the community faced the following question:

\begin{center}
    \emph{Do all graphs admit $+k$ additive spanners of near-linear size, with $k$ a constant, or at least a small enough function of $n$ to be practically efficient?}
\end{center}

This problem became a central focus of the community following a sequence of upper bound results.
First was the seminal 1995 work of Aingworth, Chekuri, Indyk, and Motwani \cite{ACIM99}, which proved that every $n$-node graph has a $+2$ additive spanner on $O(n^{3/2})$ edges.
This was followed by a theorem of Chechik \cite{Chechik13soda} that all graphs have $+4$ spanners on $\Oish(n^{7/5})$ edges (see also \cite{AlDhalaan21}), and from Baswana, Kavitha, Mehlhorn, and Pettie \cite{BKMP10} who proved that all graphs have $+6$ spanners on $O(n^{4/3})$ edges (see also \cite{Woodruff10, Knudsen14}).
Thus, it seemed that one might be able to continue this trend, paying more and more constant additive error in exchange for sparser and sparser spanners.
Unfortunately, a barrier to further progress was discovered by Abboud and Bodwin \cite{AB17jacm}, who constructed graphs for which any additive spanner on $O(n^{4/3 - c})$ edges suffers $+n^{\Omega(1)}$ additive error.
Thus, near-linear spanners with subpolynomial error are not generally possible.

That said, the lower bound of \cite{AB17jacm} is a small enough polynomial to be entirely practical even on enormous input graphs.
This work showed graphs for which any $O(n)$-size spanner pays at least $+n^{1/22}$ error, which is a small enough polynomial to be entirely practical even on huge input graphs.
However, the upper bounds for $O(n)$-size spanners were far from this ideal.
The first nontrivial constructions of $O(n)$-size spanners \cite{Pettie09} (following \cite{BCE05}) had $+n^{9/16}$ error.
Thus began a focused effort by the community to narrow these upper and lower bounds on additive error towards the middle, with the goal to determine whether practically-significant constructions of near-linear additive spanners can be generally achieved.

Despite a high throughput of recent work on the topic, the attainable error bounds for linear-size additive spanners remain wide open.
The initial lower bound of $+n^{1/22}$ additive error was improved to $+n^{1/11}$ in two concurrent papers \cite{HP18, Lu19}, and then recently to $+n^{1/10.5}$ by Lu, Vassilevska Williams, Wein, and Xu \cite{LVWX22} where it stands today.
Meanwhile, the upper bound on additive error was improved to $+\Oish(n^{1/2})$ \cite{BV15} and then to $+n^{3/7 + \eps}$ \cite{BV16} (see also \cite{Chechik13soda, Pettie09}).

\subsection{New Lower Bounds}

The main results of this paper are two new lower bound for additive spanners:
\begin{theorem} [First Main Result]
There are $n$-node graphs that do not admit a $+o(n^{1/7})$ additive spanner on $O(n)$ edges.
\end{theorem}

Our construction follows the \emph{obstacle product framework} used to prove lower bounds in prior work, but with a key generalization.
Previously, the obstacle product carefully composes a \emph{layered outer graph} with a set of \emph{layered inner graphs} in a way that causes the shortest paths in the composed graphs to inherit desirable structural properties from each.
We provide a stronger framework for analysis that allows one to compose \emph{unlayered} inner and outer graphs.

It turns out that this unlayering technique also addresses a related open problem in the area.
One can relax additive spanners to \emph{pairwise additive spanners}, where we only need to approximate distances among a set of demand pairs $P$ taken on input:
\begin{definition} [Pairwise Additive Spanners \cite{CE06, cygan2013pairwise}]
For a graph $G = (V, E)$ and a set of demand pairs $P \subseteq V \times V$, a subgraph $H \subseteq G$ is a $+k$ additive spanner of $G$ with respect to $P$ if, for all $(s, t) \in P$, we have $\dist_H(s, t) \le \dist_G(s, t) + k$.
\end{definition}
We can then hope for better error bounds in the setting where the number of demand pairs $|P|$ is not too large.
Pairwise spanners were introduced by Coppersmith and Elkin \cite{CE06}, who specifically focused on the \emph{exact} $k=0$ case (also called \emph{distance preservers}); the \emph{approximate} $k > 0$ case has been studied repeatedly in followup work \cite{abboud2016error, kavitha2017new, kavitha2013small, doi:10.1137/140953927, Knudsen14, cygan2013pairwise}.
The initial paper by Coppersmith and Elkin \cite{CE06} established the following fundamental result:
\begin{theorem} [\cite{CE06}]~
\begin{itemize}
\item (Upper Bound) For any $n$-node graph $G = (V, E)$ and set $P \subseteq V \times V$ of $|P| = O(\sqrt{n})$ demand pairs, there is a distance preserver ($+0$ pairwise spanner) on $O(n)$ edges.

\item (Lower Bound) For any $p = \omega(\sqrt{n})$, there are $n$-node graphs $G = (V, E)$ and sets of $|P|=p$ demand pairs that do not admit a distance preserver on $O(n)$ edges.
\end{itemize}
\end{theorem}

Thus, with a budget of $O(n)$ edges, we can exactly preserve distances among $O(\sqrt{n})$ demand pairs and no more.
A question asked repeatedly in followup work \cite{CE06, cygan2013pairwise, abboud2016error} is whether this $\sqrt{n}$ threshold can be improved if we allow constant $+k$ error; this question was explicitly studied in \cite{cygan2013pairwise, abboud2016error}, without resolution.
We settle this question negatively:
\begin{theorem} [Second Main Result] \label{thm:intropairwiselb}
For any constant $k > 0$ and $p = \omega(\sqrt{n})$, there are $n$-node graphs $G = (V, E)$ and sets of $|P|=p$ demand pairs that do not admit a $+k$ pairwise spanner on $O(n)$ edges.
\end{theorem}

On a technical level, this stronger lower bound is again proved using an unlayered obstacle product.
In our view, our new lower bound further cements the importance of the $\sqrt{n}$ demand pair/$O(n)$ size threshold by Coppersmith and Elkin: it is even robust to \emph{any constant additive error}.

\section{Technical Overview of Main Result and Comparison to Prior Work \label{sec:techoverview}}

At a technical level, our main result departs from prior work by altering a fundamentally different piece of the construction than the one typically considered previously.
We overview the construction and our improvement in the next section.
Here we overview the parts of our construction that match prior work, and we overview the new technical ingredients that give our improved lower bounds.

\subsection{The Obstacle Product Framework}

Like every other lower bound in the area, we follow the \emph{obstacle product} framework.
This framework involves composing an \emph{outer graph} and many copies of an \emph{inner graph}.

The essential property of the outer graph $G_O$ is that it has a set of \emph{critical pairs} $P_O \subseteq V(G_O) \times V(G_O)$, such that:
\begin{itemize}
    \item for each critical pair $(s, t)$ there is a unique shortest $s \leadsto t$ path $\pi(s, t)$ in $G_O$, called the \emph{canonical path}, and 
    \item the canonical paths are pairwise edge-disjoint.
\end{itemize}

If we remove an edge from a canonical path $\pi(s, t)$, then since it is a unique shortest path, $\dist(s, t)$ will increase by at least $+1$.
To amplify this error, we apply an \emph{edge-extension step} in which we add $k=n^{\Omega(1)}$ new nodes along every edge.
Thus, removing an edge from an (edge-extended) canonical path $\pi(s, t)$ will cause $\dist(s, t)$ to increase by at least $+k$.

\begin{figure} [H]
\begin{center}
    \begin{tikzpicture}
    \draw [fill=black] (0, 0) circle [radius=0.15];
    \draw [fill=black] (6, 0) circle [radius=0.15];
    \draw [ultra thick] (0, 0) -- (6, 0);
    
    \node at (0, -0.5) {$s$};
    \node at (6, -0.5) {$t$};
    \node at (3, -0.5) {$\pi(s, t)$};
    
    \draw [fill=black] (2, 0) circle [radius=0.1];
    \draw [fill=black] (2.5, 0) circle [radius=0.1];
    \node [red] at (2.25, 0) {\Huge $\times$};
    
    \draw [fill=black] (-1, 1) circle [radius=0.15];
    \draw [ultra thick, gray] (0, 0) to (-1, 1) to[bend left=10] (6, 0);
    \node [gray, rotate=-5] at (3, 1.2) {$\dist(s, t) \ge |\pi(s, t)| + 1$};
    
    \draw [ultra thick, ->] (6.5, 0.5) -- (7.5, 0.5);
    
\begin{scope}[shift={(9,0)}]

\draw [fill=black] (0, 0) circle [radius=0.15];
    \draw [fill=black] (6, 0) circle [radius=0.15];
    \draw [ultra thick] (0, 0) -- (6, 0);
    
    \node at (0, -0.5) {$s$};
    \node at (6, -0.5) {$t$};
    \node at (3, -0.5) {$\pi(s, t)$};
    
    \draw [fill=black] (2, 0) circle [radius=0.1];
    \draw [fill=black] (2.5, 0) circle [radius=0.1];
    \node [red] at (2.25, 0) {\Huge $\times$};
    
    \draw [fill=black] (-1, 1) circle [radius=0.15];
    \draw [ultra thick, gray] (0, 0) to (-1, 1) to[bend left=10] (6, 0);
    \node [gray, rotate=-5] at (3, 1.2) {$\dist(s, t) \ge |\pi(s, t)| + k$};
    
    \draw [fill=gray] (-0.25, 0.25) circle [radius=0.1];
    \draw [fill=gray] (-0.5, 0.5) circle [radius=0.1];
    \draw [fill=gray] (-0.75, 0.75) circle [radius=0.1];
    
    \draw [fill=gray] (0.4, 0) circle [radius=0.1];
    \draw [fill=gray] (0.8, 0) circle [radius=0.1];
    \draw [fill=gray] (1.2, 0) circle [radius=0.1];
    \node [gray] at (1.6, 0.2) {$\cdots$};
    
    \draw [fill=gray] (-0.6, 1) circle [radius=0.1];
    \draw [fill=gray] (-0.2, 1) circle [radius=0.1];
    \draw [fill=gray] (0.2, 1) circle [radius=0.1];
    \node [gray] at (0.6, 1.2) {$\cdots$};

\end{scope}
    \end{tikzpicture}
\end{center}
    
    \caption{Deleting an edge from a canonical path stretches its distance by at least $+1$.  After the \emph{edge extension step}, deleting an edge stretches distance by at least $+k$.}
\end{figure}
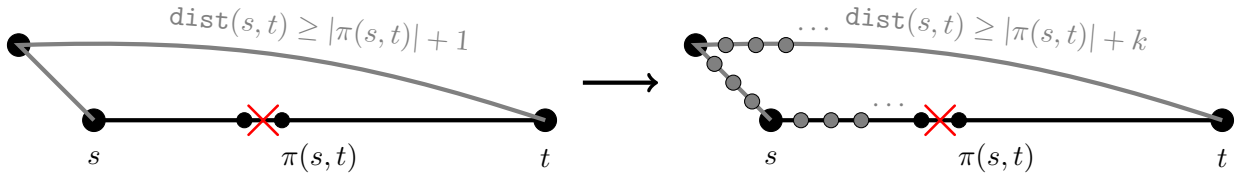

Although edge-extension forces any nontrivial spanner of the outer graph to suffer $+k$ distance error, the problem is that the extended outer graph is now very sparse: the vast majority of the nodes are new edge-extension nodes of degree $2$, while only a small handful of nodes are from the original outer graph and may have higher degree.
Thus the trivial spanner, that keeps the entire outer graph, has $O(n)$ size (where $n$ is the number of nodes after edge extension) and can be used.
The purpose of the next \emph{inner graph replacement step} in the obstacle product is to regain this lost density, so that an $O(n)$-size spanner actually has to remove a significant number of edges from the graph.

The inner graph $G_I$ is equipped with a set of critical pairs $P_I$ with unique edge-disjoint shortest paths, just like the outer graph.
We will call these \emph{inner-canonical paths}, and the paths in the outer graph \emph{outer-canonical paths} to make clear the distinction.
Let $v$ be a node in the outer graph that is contained in exactly $d$ outer-canonical paths, and suppose the inner graph $G_I$ has exactly $|P_I|=d$ critical pairs (the case where these quantities only \emph{approximately} match, instead of both being exactly $d$, can be handled easily).
We then replace the node $v$ with a copy of the entire graph $G_I$.
To preserve outer-canonical paths, we associate each outer-canonical path $\pi(s, t)$ that intersects $v$ to some inner-canonical path $\pi(s_I, t_I)$ in $G_I$.
We arrange the composed graph in such a way that the unique shortest $s \leadsto t$ path in the composed graph is exactly the original outer-canonical path $\pi(s, t)$, with each node $v$ replaced by the corresponding inner-canonical path $\pi(s_I, t_I)$ in the copy of $G_I$ that replaced $v$.
We call this unique shortest $s \leadsto t$ path the \emph{composed-canonical path}.

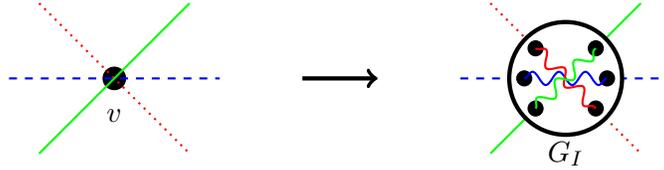
\begin{figure} [H]
\begin{center}
\begin{tikzpicture}
    \draw [fill=black] (0, 0) circle [radius=0.15];
    \draw [thick, dotted, red] (-1, 1) -- (1, -1);
    \draw [thick, dashed, blue] (-1.4, 0) -- (1.4, 0);
    \draw [thick, green] (-1, -1) -- (1, 1);
    \node at (0, -0.5) {$v$};
    
    \draw [ultra thick, ->] (2.5, 0) -- (3.5, 0);
    
\begin{scope}[shift={(6,0)}]
    
    \draw [ultra thick, black, fill=white] (0, 0) circle [radius=0.5];
    \draw [thick, dotted, red] (-1, 1) -- (1, -1);
    \draw [thick, dashed, blue] (-1.4, 0) -- (1.4, 0);
    \draw [thick, green] (-1, -1) -- (1, 1);
    \draw [ultra thick, black, fill=white] (0, 0) circle [radius=0.75];
    \node at (0, -1) {$G_I$};
    
    \draw [fill=black] (-0.55, 0) circle [radius=0.1];
    \draw [fill=black] (0.55, 0) circle [radius=0.1];
    
    \draw [fill=black] (-0.4, 0.4) circle [radius=0.1];
    \draw [fill=black] (0.4, 0.4) circle [radius=0.1];
    
    \draw [fill=black] (-0.4, -0.4) circle [radius=0.1];
    \draw [fill=black] (0.4, -0.4) circle [radius=0.1];
    
    \draw [thick, blue, snake it] (-0.55, 0) -- (0.55, 0);
    \draw [thick, red, snake it] (-0.4, 0.4) -- (0.4, -0.4);
    \draw [thick, green, snake it] (-0.4, -0.4) -- (0.4, 0.4);

\end{scope}
\end{tikzpicture}
\end{center}
    
    \caption{To regain lost density from edge extension, the \emph{inner graph replacement step} replaces nodes in the outer graph with copies of the inner graph, carefully attaching canonical paths in the outer graph to canonical paths in the inner graph.}
\end{figure}

Most of the edges in the composed graph lie in inner graphs.
This implies that, in any $O(n)$-size spanner $H$ of the composed graph $G_O \otimes G_I$, there will be a composed-canonical path $\pi(s, t)$ where the spanner is missing most of the edges used by $\pi(s, t)$ in inner graphs.
We use two cases to argue that $\dist_H(s, t)$ must be much longer than $\dist_G(s, t)$.
In the first case, suppose that the shortest $s \leadsto t$ path in $H$ uses the same sequence of inner graphs as the composed-canonical path.
It then suffers at least $+1$ error in \emph{most} of its inner graphs due to missing edges; it intersects $k = n^{\Omega(1)}$ total inner graphs, for a total cost of $+\Theta(k)$ additive error.
The other case is when the shortest $s \leadsto t$ path in the spanner avoids these gutted inner graphs by instead following a path that corresponds to a non-canonical $s \leadsto t$ path in the outer graph.
This other kind of path also suffers $+\Theta(k)$ error over the composed-canonical path, essentially due to the edge-extension step.

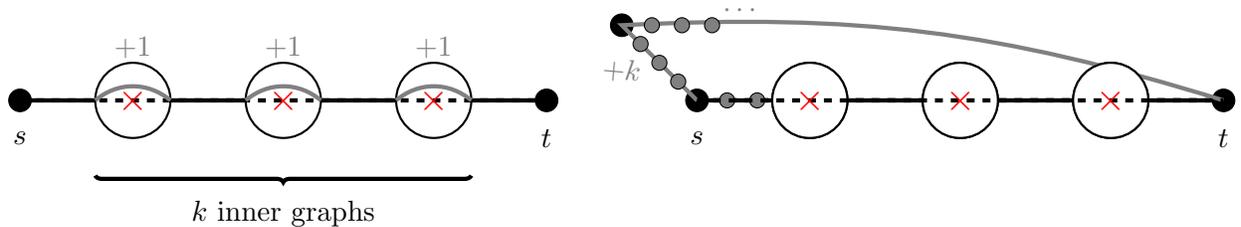
\begin{figure} [h]
    \begin{center}
    \begin{tikzpicture}
    
    \draw [fill=black] (0, 0) circle [radius=0.15];
    \draw [fill=black] (7, 0) circle [radius=0.15];
    \draw [ultra thick] (0, 0) -- (7, 0);
    
    \node at (0, -0.5) {$s$};
    \node at (7, -0.5) {$t$};
    
    \draw [black, thick, fill=white] (1.5, 0) circle [radius=0.5];
    \draw [black, thick, fill=white] (3.5, 0) circle [radius=0.5];
    \draw [black, thick, fill=white] (5.5, 0) circle [radius=0.5];

    \draw [ultra thick, dashed] (0, 0) -- (7, 0);
    \node [red] at (1.5, 0) {\Large $\times$};
    \node [red] at (3.5, 0) {\Large $\times$};
    \node [red] at (5.5, 0) {\Large $\times$};
    
    \draw [ultra thick, gray] (1, 0) to[bend left=40] (2, 0);
    \draw [ultra thick, gray] (3, 0) to[bend left=40] (4, 0);
    \draw [ultra thick, gray] (5, 0) to[bend left=40] (6, 0);
    \node [gray] at (1.5, 0.7) {$+1$};
    \node [gray] at (3.5, 0.7) {$+1$};
    \node [gray] at (5.5, 0.7) {$+1$};

    \draw [thick, decorate, decoration={brace}, ultra thick] (6, -1) -- (1, -1);
    \node at (3.5, -1.5) {$k$ inner graphs};

\begin{scope}[shift={(9,0)}]

    \draw [fill=black] (0, 0) circle [radius=0.15];
    \draw [fill=black] (7, 0) circle [radius=0.15];
    \draw [ultra thick] (0, 0) -- (7, 0);
    
    \node at (0, -0.5) {$s$};
    \node at (7, -0.5) {$t$};
    
    \draw [black, thick, fill=white] (1.5, 0) circle [radius=0.5];
    \draw [black, thick, fill=white] (3.5, 0) circle [radius=0.5];
    \draw [black, thick, fill=white] (5.5, 0) circle [radius=0.5];

    \draw [ultra thick, dashed] (0, 0) -- (7, 0);
    \node [red] at (1.5, 0) {\Huge $\times$};
    \node [red] at (3.5, 0) {\Huge $\times$};
    \node [red] at (5.5, 0) {\Huge $\times$};
    
    \draw [fill=black] (-1, 1) circle [radius=0.15];
    \draw [ultra thick, gray] (0, 0) to (-1, 1) to[bend left=10] (7, 0);
    \node [gray, rotate=-5] at (-1, 0.4) {$+ k$};
    
    \draw [fill=gray] (-0.25, 0.25) circle [radius=0.1];
    \draw [fill=gray] (-0.5, 0.5) circle [radius=0.1];
    \draw [fill=gray] (-0.75, 0.75) circle [radius=0.1];
    
    \draw [fill=gray] (0.4, 0) circle [radius=0.1];
    \draw [fill=gray] (0.8, 0) circle [radius=0.1];
    \draw [fill=gray] (1.2, 0) circle [radius=0.1];
    \node [gray] at (1.6, 0.2) {$\cdots$};
    
    \draw [fill=gray] (-0.6, 1) circle [radius=0.1];
    \draw [fill=gray] (-0.2, 1) circle [radius=0.1];
    \draw [fill=gray] (0.2, 1) circle [radius=0.1];
    \node [gray] at (0.6, 1.2) {$\cdots$};
    
    \draw [black, thick, fill=white] (1.5, 0) circle [radius=0.5];
    \draw [black, thick, fill=white] (3.5, 0) circle [radius=0.5];
    \draw [black, thick, fill=white] (5.5, 0) circle [radius=0.5];

    \draw [ultra thick, dashed] (0, 0) -- (7, 0);
    \node [red] at (1.5, 0) {\Large $\times$};
    \node [red] at (3.5, 0) {\Large $\times$};
    \node [red] at (5.5, 0) {\Large $\times$};

\end{scope}

    \end{tikzpicture}
    \end{center}
    \caption{In any $O(n)$-size spanner $H$, there is a canonical path $\pi(s, t)$ that is missing most of its edges in inner graphs.  The additive error of $s \leadsto t$ paths in the spanner are analyzed in two types: those that suffer $+1$ error in each inner graph (left), and those that avoid the problematic inner graphs entirely by following a non-outer-canonical path (right).}
\end{figure}

\subsection{Improvements in Prior Work: Changes to the Alternation Product}

Essentially every major improvement in the lower bound has thusfar been achieved by an improvement to the \emph{alternation product}.
Although the alternation product is not used at all in the technical part of this paper, it is worth overviewing, to highlight the core difference between the new improvements in this paper and those obtained in prior work.

The motivating observation behind the alternation product is that, for correctness of the lower bound, one needs that each edge in an inner graph is only used by one composed-canonical path.
To enforce this, it is actually \emph{overkill} to require edge-disjoint outer-canonical paths.
Rather, we can allow the pairwise intersections of outer-canonical paths to contain at most one edge.

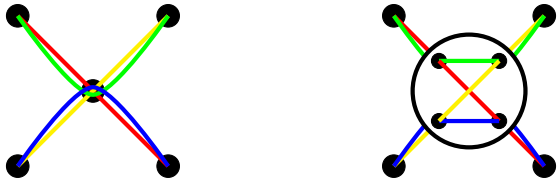
\begin{figure} [H]
\begin{center}
\begin{tikzpicture}
    \draw [fill=black] (0, 0) circle [radius=0.15];
    \draw [fill=black] (-1, 1) circle [radius=0.15];
    \draw [fill=black] (-1, -1) circle [radius=0.15];
    \draw [fill=black] (1, 1) circle [radius=0.15];
    \draw [fill=black] (1, -1) circle [radius=0.15];
    
    \draw [ultra thick, red] (-1, 1) -- (1, -1);
    \draw [ultra thick, yellow] (-1, -1) -- (1, 1);
    \draw [ultra thick, green] plot[smooth] coordinates {(-1, 1) (0, -0.05) (1, 1)};
    \draw [ultra thick, blue] plot[smooth] coordinates {(-1, -1) (0, 0.05) (1, -1)};
    
\begin{scope} [shift={(5, 0)}]
    \draw [fill=black] (-1, 1) circle [radius=0.15];
    \draw [fill=black] (-1, -1) circle [radius=0.15];
    \draw [fill=black] (1, 1) circle [radius=0.15];
    \draw [fill=black] (1, -1) circle [radius=0.15];
    
    \draw [ultra thick, red] (-1, 1) -- (1, -1);
    \draw [ultra thick, yellow] (-1, -1) -- (1, 1);
    \draw [ultra thick, green] plot[smooth] coordinates {(-1, 1) (0, -0.05) (1, 1)};
    \draw [ultra thick, blue] plot[smooth] coordinates {(-1, -1) (0, 0.05) (1, -1)};
    
    \draw [ultra thick, black, fill=white] (0, 0) circle [radius=0.75];
    
    \draw [fill=black] (-0.4, -0.4) circle [radius=0.1];
    \draw [fill=black] (-0.4, 0.4) circle [radius=0.1];
    \draw [fill=black] (0.4, -0.4) circle [radius=0.1];
    \draw [fill=black] (0.4, 0.4) circle [radius=0.1];
    
    \draw [ultra thick, green] (-0.4, 0.4) -- (0.4, 0.4);
    \draw [ultra thick, red] (-0.4, 0.4) -- (0.4, -0.4);
    \draw [ultra thick, blue] (-0.4, -0.4) -- (0.4, -0.4);
    \draw [ultra thick, yellow] (-0.4, -0.4) -- (0.4, 0.4);
\end{scope}

\end{tikzpicture}
    \end{center}
    \caption{We can allow the outer-canonical paths to contain an edge in their pairwise intersections (left), and then after the inner graph replacement step, the composed-canonical paths will become edge-disjoint on inner graphs (right).}
\end{figure}

Thus, there is some additional freedom in outer graph design.
The alternation product tries to leverage this freedom for a stronger lower bound.
It changes the edge-disjoint outer-canonical paths into $2$-path-disjoint outer-canonical paths, in exchange for different relative counts of nodes, critical pairs, and canonical path lengths in the outer graph.
These parameter changes are favorable when the goal is to build a lower bound \emph{against denser spanners}; for example, the alternation product is necessary to establish the existence of graphs that need $+n^{\Omega(1)}$ error for any spanner on $O(n^{4/3 - \eps})$ edges.
But, in a naive implementation of the alternation product, these parameter changes are \emph{unfavorable} when the goal is lower bounds against $O(n)$-size spanners.

The improved lower bounds of $\Omega(n^{1/11})$ from \cite{HP18, Lu19} are mostly achieved by removing the alternation product from \cite{AB17jacm} entirely.
The subsequent $\Omega(n^{1/10.5})$ lower bounds of Lu et al \cite{LVWX22} reintroduce the alternation product; their main technical innovation is a clever new implementation of the alternation product that takes advantage of the geometric structure of the outer graph to obtain more favorable parameter changes, which make it beneficial even in the setting of sparse spanners.

\subsection{Improvements in Our Work: New Inner/Outer Graphs}

The current paper obtains achievements in a different way from prior work.
We omit the alternation product entirely; in that sense, our construction forks \cite{HP18, Lu19} rather than \cite{LVWX22} (we briefly explain why in the following section).
Rather, we enable a \emph{significant technical change in the design of inner/outer graphs}, which we explain next.

Roughly, the goal of the inner/outer graphs is to pack in as many critical pairs as possible, with as long canonical paths as possible.
The main technical innovation in the $\Omega(n^{1/22})$ lower bound of \cite{AB17jacm} was to replace the ``butterfly'' outer graph implicitly used by Woodruff \cite{Woodruff06} with a ``distance preserver lower bound graph,'' along the lines of a construction by Coppersmith and Elkin \cite{CE06} (see also \cite{Alon02, Hesse03} for prior graph constructions based on a similar technique).
Ideally, one would like to use the Coppersmith-Elkin distance preserver lower bound graphs \emph{exactly} for inner/outer graphs.
But there's a catch.
When we execute the inner graph replacement step, we need to make sure that the composed-canonical paths are unique shortest paths in the composed graph.
This property is not immediate, and in fact it does not hold for arbitrary choices of inner/outer graphs.
Rather, it holds if the inner and outer graphs are both \emph{layered}.
But layeredness is not free; introducing layeredness to the Coppersmith-Elkin construction significantly harms the inner/outer graph quality, leading to worse lower bounds.
All previous lower bound constructions \cite{Woodruff06, AB17jacm, Lu19, HP18, LVWX22} pay this penalty in order to layer their graphs.
\begin{center}
    \begin{tabular}{|lllll|}
        \hline
        Lower Bound & Outer Graph & Inner Graph & Alternation Product? & Citation \\
        \hline
        $\Omega(\log n)$ & Butterfly & Biclique & No & \cite{Woodruff06}\\
        $\Omega(n^{1/22})$ & Layered Dist Pres LB & Biclique & Yes & \cite{AB17jacm}\\
        $\Omega(n^{1/11})$ & Layered Dist Pres LB & Layered Dist Pres LB & No & \cite{HP18, Lu19}\\
        $\Omega(n^{1/10.5})$ & Layered Dist Pres LB & Biclique & Improved & \cite{LVWX22}\\
        $\Omega(n^{1/7})$ & Dist Pres LB & Dist Pres LB & No & this paper\\
        \hline
    \end{tabular}
\end{center}






 The technical contribution of our paper can be summarized as follows:
\begin{lemma} [Main Technical Lemma, Informal]
There are certain amended versions of the Coppersmith-Elkin graphs in \cite{CE06} whose structure allows them to be used as inner/outer graphs in the obstacle product, despite being unlayered.
\end{lemma}

In addition to its use in spanner lower bounds, this technical lemma is also the essential missing ingredient towards our extension of pairwise distance preserver lower bounds of Coppersmith and Elkin to pairwise additive spanners with $+k$ error (Theorem \ref{thm:intropairwiselb}).
Our lower bound matches the $\sqrt{n}$ demand pair threshold obtained by the Coppersmith-Elkin lower bound  against distance preservers precisely because we can use an amended Coppersmith-Elkin distance preserver lower bound for the outer graph of our obstacle product.

One can view the Coppersmith-Elkin graph construction as parametrized by a \emph{convex set of vectors} taken on input.
The original graphs in \cite{CE06} use a standard convex set construction from \cite{BL98}.
We need to design very precise convex sets that have essentially the same size as those used by Coppersmith and Elkin, but with some additional technical properties that enable use with the obstacle product.
Some of our main new technical contributions lie in the design of these convex sets, which we describe in Section \ref{sec:specifying}.

The other major technical step in this paper lies in the part of the obstacle product where outer-canonical paths are attached to inner-canonical paths in the inner graph replacement step.
In all prior work, it has been completely arbitrary which outer-canonical path was attached to which inner-canonical path.
In this work, it is non-arbitrary: we show that the obstacle product benefits from a specific \emph{alignment} between these paths; roughly, outer-canonical paths are attached to inner-canonical paths based on the direction they are travelling. 
Details are given in Section \ref{sec:specifyingfinal}.

With this, we employ a more complex version of the error analysis from prior work, that leverages our convex set design and alignment between inner and outer canonical paths.
We introduce a \emph{move decomposition framework} to do so, which enables an amortized version of the convexity argument used for error analysis in prior work. A high-level overview of the move decomposition framework can be found in Section \ref{sec:analysis_1}, but we will not overview it further here.

\subsection{Future Directions: Can the Obstacle Product Achieve Tight Error Bounds?}

This research project was initiated by a thought experiment: is it even \emph{conceivable} for the obstacle product framework to produce lower bounds matching the upper bounds obtained by the \emph{path-buying} framework currently used for the spanner upper bounds in this paper and in \cite{BKMP10, BV16}?
In this subsection, we argue that the answer is a resounding ``sort of:'' on all axes except one, there is clear potential for these frameworks to produce matching upper and lower bounds.
This problematic axis likely spoils the possibility of pinning down an exact error bound for $O(n)$-size spanners in the near future, but this problematic axis is more of a barrier in analysis than in construction.

Suppose we run the current state-of-the-art upper bound construction from \cite{BV16} on lower bound graphs produced by the obstacle product.
One arrives at three main points of technical disagreement, where the upper bound analysis makes pessimistic assumptions not actually realized on the current lower bound graph.
Thus, a hypothetical tight analysis would have to either introduce a lower bound graph that realizes these pessimistic assumptions, or it would need to improve the upper bound argument to avoid making these pessimistic assumptions in the first place.
We overview these points, and their implications for future work, below.

\paragraph{Our Unlayered Inner Graphs Are Probably Necessary.}

The path-buying framework used in previous spanner upper bounds \cite{AB16soda, Knudsen14, BKMP10}  begins with a \emph{clustering} phase,  in which the graph nodes are partitioned into low-radius ``clusters.''
In the second \emph{path-buying phase}, one adds a collection of shortest paths that connect far-away clusters with small additive error.
A key piece of the upper bound analysis argues that we only add a small number of shortest paths through each cluster.
More specifically, leveraging distance preserver upper bounds from \cite{CE06} for a cluster $C$ with $n_C$ nodes, we can afford $O(n_C^{1/2})$ shortest paths while paying only a linear number of edges for this cluster $C$.

When the path-buying framework is run on an obstacle product construction, it precisely picks out the inner graphs (plus a few nodes in the attached edge-extended paths) as clusters, and it picks out the outer-canonical paths as the shortest paths added in the second phase.
Thus, for tightness, one would hope that parameters balance in such a way that we have $\Omega(n_I^{1/2})$ canonical paths through each inner graph on $n_I$ nodes.
Prior to this work, this was not so: the forced layered structure of the inner graphs meant that one actually had to place $\Omega(n_I^{2/3})$ canonical paths \cite{HP18} or even $\Omega(n_I)$ canonical paths \cite{AB17jacm, LVWX22} to achieve a lower bound.
By unlayering our inner graphs, we are able to rebalance parameters to have $\Omega(n_I^{1/2})$ canonical paths through each inner graph for the first time.
Thus, this particular axis is no longer a point of disagreement following our work; we view our main conceptual contribution as demonstrating tightness between the path buying and obstacle product frameworks in this regard.

\paragraph{An Improved Alternation Product Is Probably Still Needed.}
In the clustering phase of the path-buying  framework, each cluster is either ``small'' or ``large,'' depending on its number of nodes. The worst-case input graphs for the spanner upper bounds would have the structure that all clusters are right on the small/large borderline; it is a good case when all clusters are significantly above or below this threshold.
As mentioned, when one clusters obstacle product graphs, the clusters are precisely the inner graphs plus some of their attached edge-extended paths. However, they would specifically be classified as \emph{large} clusters in the upper bound, far away from this borderline. So the upper bound makes a pessimistic assumption of borderline clusters that is not actually realized in the current lower bound construction.

Let us engage for a moment in some wishful thinking.
Suppose that we could apply an alternation product on the outer graph, and then replace in inner graphs with our current density of canonical paths but with the additional structure $S \times S$ on their demand pairs (such graphs are constructed in \cite{CE06}).
This would substantially reduce the number of edge-extended paths attached to each inner graph, and this change would put our inner graphs right on the small/large borderline when viewed as clusters.
Thus, to resolve this discrepancy between the path buying and obstacle product frameworks, we think that the alternation product or something similar is very likely to be the right tool.

There is no intrinsic barrier to applying an alternation product on top of our unlayering method, but it complicates the already-delicate geometric details of our argument in a way that we have not resolved.
So we wish to emphasize that the improved alternation product in \cite{LVWX22} remains an important technical idea, and the natural next step for the area is to integrate this alternation product (or perhaps one with even further-improved parameters) with our unlayered inner graphs.
In this sense, although our lower bounds improve numerically on \cite{LVWX22}, we think it is more conceptually correct to consider our respective constructions as concurrent state-of-the-art that achieve two different desirable features of the ultimate lower bound construction, which will need to be unified in future work.

\paragraph{Optimal Outer Graphs Will Probably Be Hard To Achieve.}

In order to replace in inner graphs of nontrivial size, we need to start the obstacle product construction with a relatively dense outer graph, that has $\poly(n)$ canonical paths passing through a typical node.
Such an outer graph would essentially need to be a distance preserver lower bound as in \cite{CE06} (perhaps passed through an obstacle product).
The trouble is that it is arguably out of reach of current techniques to determine the optimal quality of a distance preserver lower bound graph.
Distance preserver lower bounds have close relationships to several long-standing open problems in extremal combinatorics, like bounds for the triangle removal lemma \cite{Bodwin21}, and it will probably be difficult to settle the bounds for distance preservers without also making a breakthrough on these difficult combinatorial problems.

This paper is the first that is able to use state-of-the-art distance preserver lower bounds for the outer graph.
One can explain a substantial part of the remaining numerical upper/lower bound gap for additive spanners by acknowledging that the upper bound implicitly uses off-the-shelf distance preserver upper bounds, and the lower bound uses distance preserver lower bounds, and these are far apart.
Thus: while it is quite reasonable to think that the obstacle product might produce tight lower bounds \emph{when an optimal distance preserver lower bound outer graph is plugged in}, it will likely be hard to find a concrete graph to plug in.






\section{Construction Framework}
We now present our lower bound construction.
We refer back to the technical overview (Section \ref{sec:techoverview}) for intuition, comparison to prior work, and to highlight the part of this paper that is new.
For simplicity of presentation, we will frequently ignore issues related to non-integrality of expressions that arise in our analysis; these issues affect our bounds only by lower-order terms.

\subsection{Base Graph $G_B$}
We start by describing a template for a \emph{base graph} $G_B$, which is a generalized version of the lower bound construction for distance preservers by Coppersmith and Elkin \cite{CE06} (they use the following construction with a particular choice of $W$).
The outer and inner graphs in our obstacle product will both be versions of the base graph, instantiated a bit differently. Graph $G_B$ will have parameters $(x, y, r, W)$. 

\paragraph{Vertices.} The vertices of the base graph are $[1, x] \times [1, y]$, where $x, y$ are positive integers that are inputs to the construction.
We imagine these vertices as a subset of $\zz^2$, i.e., embedded as points in the Euclidean plane.

\paragraph{The Strongly Convex Set $W$.}
The edges and critical pairs of the base graph both depend on an additional input $W$, which is required to be a \emph{strongly convex set} of vectors in $\zz^2$.
We recall the definition:
\begin{definition}[Strongly Convex Set]
A set of vectors $W$ is strongly convex if the equation $\vec{v}_0 = \lambda_1\vec{v}_1 + \cdots + \lambda_k\vec{v}_k$ has no nontrivial solutions with all $\vec{v}_i \in W$, $k$ any positive integer, and $\lambda_i$ (possibly negative) scalars with $\sum_i|\lambda_i| \leq 1$. The trivial solutions are when $\vec{v}_0 = \vec{v}_1 = \cdots = \vec{v}_k$.
\end{definition}

We will write our strongly convex set as $W(r, \psi)$ to mean that (1) the $x$-coordinate of all vectors is between $r/2$ and $r$, and (2) the angle between any vector and the horizontal is in the range $[0, \psi]$ radians.
For a technical reason to follow, we further require that the parameter $r$ satisfies $r \leq \frac{x}{4}$.

\paragraph{Critical Pairs.}

We next define the set of critical pairs $P$ for the base graph.
Let $r$ be an integer parameter of $G_B$ and $W(r, \psi)$ be our chosen strongly convex set. 
Let $S = [1, r/2] \times [1, y/2]$, and let $T = [x - r, x] \times [1, y]$.
The critical pairs $P$ are a subset of $S \times T$.
Specifically: for each $s \in S$ and each $\vec{v} \in W(r, \psi)$, let $t = s + k\vec{v}$ where $k$ is the largest positive integer such that $t \in V$, and include $(s, t) \in P$.
We quickly confirm that this node $t$ is well-defined:
\begin{lemma}
If we choose $\psi$ so that  $0 \le \psi \le \pi/4$  and $\tan \psi \leq yx^{-1} / 2$, then for all $s \in S, \vec{v} \in W(r, \psi)$ there exists a positive integer $k$ with $t := s + k\vec{v} \in T$.
\end{lemma}
\begin{proof} 
Choose $k$ to be the largest integer such that $s_1 + kv_1 \le x$, where $s_1, v_1$ are the first coordinates of $s, \vec{v}$ respectively. 
Since $\|\vec{v}\| \le r$, this implies $s_1 + kv_1 \ge x-r$, and so the first coordinate of $s + k\vec{v}$ is in the appropriate range $[x-r, x]$.

For the second coordinate: since $s_2 \ge 1$ and $\tan \psi \ge 0$, we have $s_2 + kv_2 \ge 1$.
Additionally, since $s_2 \le y/2$ and $\tan \psi \le yx^{-1}/2$, we have $s_2 + kv_2 \le y/2 + x \cdot yx^{-1}/2 \le y$.
Thus the second coordinate of $s + k\vec{v}$ is in the appropriate range $[1, y]$, completing the proof that $s + k\vec{v} \in T$.
\end{proof}
This lemma imposes a mild constraint on our choice of $\psi$, which will be satisfied in the instantiation of our inner and outer graphs from this base graph. 
 
\paragraph{Edges and Canonical Paths.}

Each critical pair $(s, t) \in P$ is generated using a vector $\vec{v} \in W(r, \psi)$; we call this the \emph{canonical vector} of $(s, t)$.
We define the \emph{canonical $(s, t)$-path $\pi_B^{s, t}$} the $(s \leadsto t)$-path containing exactly the edges of the form $(s + i\vec{v}, s+ (i+1)\vec{v})$ for integers $0 \le i < k$. 
The edges of the graph are exactly those contained in any canonical path.

\begin{figure} [h]
    \centering
    \begin{tikzpicture}
    
    \draw [black, fill=gray] (0.5, 0.5) rectangle (2.5, 5.5);
    \node at (1.5, 0) {\Huge $S$};
    \draw [black, fill=gray] (6.5, 0.5) rectangle (10.5, 10.5);
    \node at (8.5, 0) {\Huge $T$};
    
    \foreach \x in {1,...,10} {
        \foreach \y in {1,...,10} {
            \draw [fill=black] (\x, \y) circle [radius=0.1];
        }
    }
    
    \draw[fill=red!50!white] (5, 2) arc (0:45:4);
    \draw[red!50!white, fill=red!50!white] (1, 2) -- (5, 2) -- (3.84, 4.8) -- cycle;
    
     \draw[fill=red!50!white] (3, 2) arc (0:45:2);
     \draw[red!50!white, fill=red!50!white] (1, 2) -- (3, 2) -- (2.4, 3.4) -- cycle;
     
    \draw [blue!50!white, fill=blue!50!white] (4, 2.5) ellipse (0.5 and 1);
    \draw [blue!50!white, fill=blue!50!white, rotate=45] (5, 0) ellipse (0.5 and 1);

    \draw [blue, fill=blue] (1, 2) circle [radius=0.3];
    \draw [ultra thick, ->, blue] (1, 2) -- (4, 2);
    \draw [blue, fill=blue] (4, 2) circle [radius=0.2];
    \draw [ultra thick, ->, blue] (4, 2) -- (7, 2);
    \draw [blue, fill=blue] (7, 2) circle [radius=0.2];
    \draw [ultra thick, ->, blue] (7, 2) -- (10, 2);
    \draw [blue, fill=blue] (10, 2) circle [radius=0.3];
    
    \draw [ultra thick, ->, blue] (1, 2) -- (3, 4);
    \draw [blue, fill=blue] (3, 4) circle [radius=0.2];
    \draw [ultra thick, ->, blue] (3, 4) -- (5, 6);
    \draw [blue, fill=blue] (5, 6) circle [radius=0.2];
    \draw [ultra thick, ->, blue] (5, 6) -- (7, 8);
    \draw [blue, fill=blue] (7, 8) circle [radius=0.2];
    
    \draw [ultra thick, ->, blue] (7, 8) -- (9, 10);
    \draw [blue, fill=blue] (9, 10) circle [radius=0.3];

    \draw [ultra thick, ->, blue] (1, 2) -- (4, 3);
    \draw [blue, fill=blue] (4, 3) circle [radius=0.2];
    \draw [ultra thick, ->, blue] (4, 3) -- (7, 4);
    \draw [blue, fill=blue] (7, 4) circle [radius=0.2];
    \draw [ultra thick, ->, blue] (7, 4) -- (10, 5);
    \draw [blue, fill=blue] (10, 5) circle [radius=0.3];
    
    \node [fill=red!50!white] at (2, 2.4) {$\psi$};
    \node at (2.75, 3.3) {$\frac{r}{2}$};
    \node at (4.4, 4.4) {$r$};
    
    \node at (3.5, 3.5) {$W(r, \psi)$};

    \end{tikzpicture}
    \caption{For each node $s \in S$, we use the strongly convex set $W(r, \psi)$ to generate the the nodes $t \in T$ for which $(s, t)$ is included as a critical pair, in addition to the generation of the canonical paths connecting these critical pairs.}
\end{figure}
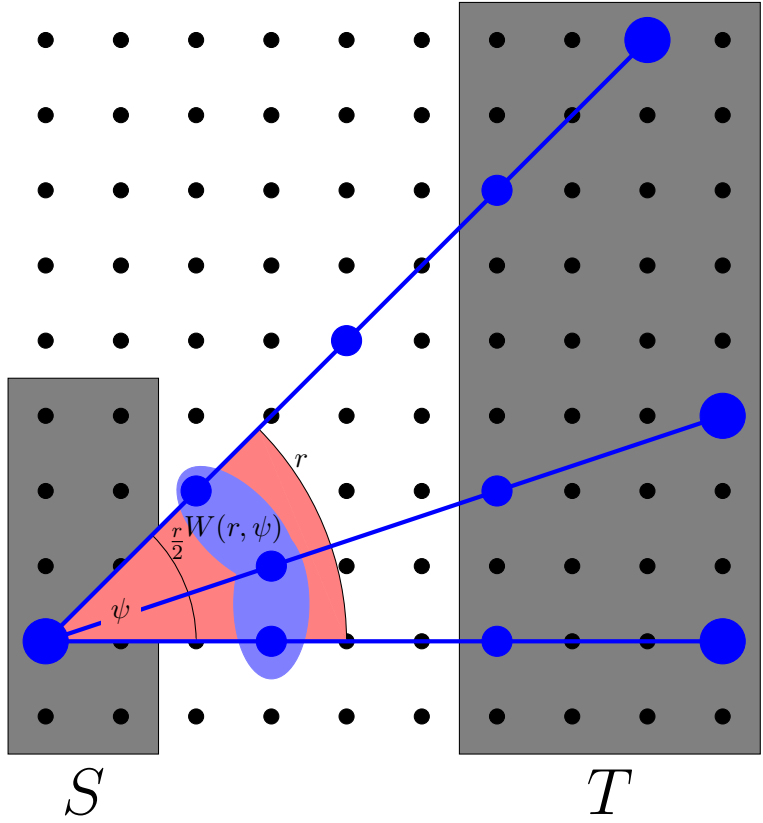



\paragraph{Important Properties of the Base Graph.}

This completes the construction of the base graph; we note its important structural properties before moving on.
A version of this lemma appears frequently in prior work.

\begin{lemma}[Properties of Base Graph $G_B$, similar to lemmas in \cite{ABP17, HP18, LVWX22}] \label{lem:basegraph}
\label{lem:base_graph}
The base graph $G_B = (V, E)$ has the following properties:
\begin{enumerate}
    \item $|V| = xy$
    \item $|P| = \Theta(ry\cdot |W(r, \psi)|)$
    \item The canonical paths $\pi_B^{s, t}$ are pairwise edge-disjoint.
    \item For all $s, t \in P$, $\left|\pi_B^{s, t}\right| = \Theta(\frac{x}{r})$.
    Consequently, $|E| = \Theta(xy \cdot |W(r, \psi)|)$.
    \item Each canonical path $\pi_B^{s, t}$ is the unique shortest $(s, t)$-path in $G_B$.
\end{enumerate}
\end{lemma}
\begin{proof} ~
\begin{enumerate}
    \item The number of vertices is immediate from construction.
    
    \item There is exactly one critical pair in $P$ for each combination of a vertex in $S$ and  vector in $W(r, \psi)$.
    Thus
    $$|P| = |S| \cdot |W(r, \psi)| = \Theta(ry \cdot |W(r, \psi)|).$$
    
    \item Each edge $(a,b) \in E$ uniquely identifies a vector $\vec{v} \in W(r, \psi)$.
    Since by construction $(a,b)$ lies on a canonical path, we can subtract $\vec{v}$ from $a$ zero or more times to reach a node in $S$; since the first coordinate of $v$ is at least $r/2$ and the width of $S$ is $r/2$, there is a \emph{unique} node in $S$ that we can reach in this way.
    Thus $(a, b)$ also uniquely identifies the first node of its canonical path $s \in S$.
    Since $s, \vec{v}$ determine a canonical path, $(a, b)$ lies on a unique canonical path.
    
    \item Let $\vec{v}$ be the canonical vector of a critical pair $(s, t)$.
    Notice that
    $$\left|\pi_B^{s, t}\right| = \frac{\|t - s\|}{\left\|\vec{v}\right\|}.$$
    Since $P \subseteq S \times T$, the first-coordinate displacement $|t_1 - s_1|$ is at least $x - 2r \geq x/2$. 
    Then since $\|\vec{v}\| \le r$, we have
    $$\left|\pi_B^{s, t}\right| = \frac{\|t - s\|}{\|\vec{v}\|} \geq \frac{x-2r}{r} \geq \frac{x}{2r}$$
    where the last inequality is since we require $r \le x/4$.
    Since critical paths are edge-disjoint and every edge lies on a critical path, it follows that
    $$|E| = |P| \cdot \Theta\left(\frac{x}{2r}\right) = \Theta(xy \cdot |W(r, \psi)|).$$
    
    \item Let $\pi_B^{s, t} = (s, s + \vec{v}, s + 2\vec{v}, \dots, t)$ be a canonical path, where $\pi_B^{s, t}$ has $k$ edges and so $t - s = k\vec{v}$.
    Suppose for the sake of contradiction that there is some other $(s, t)$-path $\pi$ of length $j \leq k$ in $G_B$, and let $\vec{v}_i$ be the vector used to create the $i^{th}$ edge in $\pi$. Then
    $$k \vec{v} = t - s = \sum_{i=1}^j \vec{v}_i.$$
  Since we have assumed that $\pi \neq \pi_B^{s, t}$, this violates strong convexity property of $W(r, \psi)$, completing the contradiction. \qedhere
\end{enumerate}
\end{proof}

\subsection{Composing the Final Graph $G$}
\paragraph{Inner Graph and Outer Graph.} We instantiate two different copies of our base graph, which we will call the \textit{inner graph} $G_I = (V_I, E_I)$ and the \textit{outer graph} $G_O = (V_O, E_O)$.
We will use subscripts $I$ and $O$ to indicate the inputs used to create $G_I, G_O$ respectively.
That is: the inner graph has dimensions $x_I, y_I$, strongly convex set $W_I(r_I, \psi_I)$, critical pairs $P_I$, and canonical paths $\pi_I^{s, t}$.
The outer graph parameters are the same with subscript $O$.
Since $G_I, G_O$ are both instantiations of the base graph, they both satisfy Lemma \ref{lem:base_graph}.

\paragraph{Inner Graph Replacement Step.} 
The next step in the construction of $G$ is to replace each  vertex in $G_O$ with a copy of the inner graph $G_I$.  
For each vertex $u$ in $G_O$ replaced with a copy $G_I^{u}$ of $G_I$, we must reconnect every edge originally incident to $u$ in $G_O$ to some vertex in $G_I^u$. Recall that the critical pair set $P_I$ is a subset of $S_I \times T_I$. We will regard the source vertices $S_I$ to be the \textit{input ports} of $G_I^u$ and the sink vertices $T_I$ to be the \textit{output ports} of $G_I^u$.
We will attach every incoming edge of form $(u - \vec{v}, u)$ in $G_O$ where $\vec{v} \in W_O$ to an input port in $G_I^u$.
Likewise, we will attach every outgoing edge of form $(u, u + \vec{v})$ in $G_O$ to an output port in $G_I^u$.  

To perform this attachment, we define a bijection $\phi: W_O \mapsto P_I$ from the vectors in the outer graph's strongly convex set to the critical pairs in $P_I$.
Later in the analysis, we will specify that $\phi$ is non-arbitrary; for technical reasons we must specifically choose a bijection satisfying certain properties.
But for now we will prove some useful properties of the construction that hold regardless of which bijection $\phi$ is used.
We note that since $\phi$ is a bijection we gain constraints
$$|W_O| = |P_I| = \Theta(r_I y_I \cdot |W_I|).$$

If $\vec{v} \in W_O$ and $\phi(\vec{v}) = (s, t)$, we plug in the incoming edge $(u - \vec{v}, u)$ into input port $s$ in $G_I^u$, and we plug in the outgoing edge of form $(u, u + \vec{v})$ originally in $G_O$ into output port $t$ in $G_I^u$.
We repeat this process for all copies of $G_I$ and all vectors $\vec{v} \in W_O$. Let $G'$ be the graph resulting from this process.

\paragraph{Edge Subdivision Step.} Let $z$ be a new parameter of the construction. To obtain our final graph $G$, we subdivide every edge of $G'$ corresponding to an original edge in outer graph $G_O$ into a path of length $z = |V_I|\cdot |P_I|^{-1}$, by adding  new nodes along the edge. 
We refer to the paths in $G$ replacing edges from $G_O$ as  \textit{subdivided paths}.

\paragraph{Critical pairs.} We define the critical pairs $P$ associated to the final graph $G$ as follows:
\begin{itemize}
    \item For each $(s_O, t_O) \in P_O$, let $G_I^{s_O}$ and $G_I^{t_O}$ be the inner graph copies in $G$ corresponding to vertices $s_O$ and $t_O$ in $G_O$.
    Let $\vec{v}_O$ be the canonical vector corresponding to critical pair $(s_O, t_O) \in P_O$, and let $\phi(\vec{v}_O) = (s_I, t_I) \in P_I$.

    \item We then add a critical pair to $P$ from the vertex $s_I$ in $G_I^{s_O}$ paired with the vertex $s_I$ in $G_I^{t_O}$.\footnote{In principle we could use $t_I$ in place of $s_I$ in $G_I^{t_o}$, but using $s_I$ instead happens to simplify some technical details later on.} Denote these vertices in $G$ as $s$ and $t$, respectively.
    
    \item The associated canonical $(s, t)$-path $\pi^{s, t}$ through the final graph $G$ is the one obtained by starting with $\pi_O^{s_O, t_O}$ and replacing each edge with the corresponding subdivided path in $G$ and each node $u$ with the canonical path $\pi_I^{s_I, t_I}$ in the graph $G_I^u$, except that we replace the final node $t_O$ with the single node $s_I$ in $G_I^{t_O}$.We define vector $\vec{v}_O \in W_O$ to be the canonical vector of $G$ associated with $\pi^{s, t}$. 
\end{itemize}

We summarize the properties of our construction:
\begin{lemma}[Properties of Final Graph $G$]
\label{lem:final_graph}
Graph $G = (V, E)$ has the following properties:
\begin{enumerate}
    \item $|V| = \Theta\left(|V_O\|V_I|\right)$
    \item $|P| = \Theta \left(r_O y_O \cdot |W_O|\right)$
    \item The canonical paths $\pi^{s, t}$ for $(s, t) \in P$ are pairwise edge-disjoint.
    \item For all $(s, t) \in P$, $\left|\pi^{s, t}\right| = \Theta\left(\frac{x_O}{r_O}\cdot \frac{x_I}{r_I}\right)$. Consequently, $|E| = \Omega(|W_I||V|)$. 
\end{enumerate}
\end{lemma}
\begin{proof} ~
\begin{enumerate}
\item The number of vertices in inner graph copies in $G$ is $|V_O\|V_I|$.
Now we just need to count the vertices in the subdivided paths of $G$.
Each inner graph copy $G_I$ in $G$ is incident to at most $2|W_O|$ subdivided paths, each of which has length $z = |V_I| \cdot |P_I|^{-1}$.
Then since $|P_I| = |W_O|$ by our bijection $\phi$, the number of vertices in subdivided paths is at most $|V_O| \cdot 2|W_O| \cdot |V_I\|W_O|^{-1} = 2|V_O\|V_I|$.

\item The number of demand pairs follows immediately from Lemma \ref{lem:basegraph} and the fact that $|P| = |P_O|$.

\item The fact that canonical paths do not share edges along subdivided paths follows from edge-disjointness of canonical paths in the outer graph (Lemma \ref{lem:basegraph}).
The fact that canonical paths do not share edges in inner graph copies follows by noticing that any two canonical paths in $G_O$ with the same canonical vector $\vec{v} \in W_O$ are node-disjoint.
Thus, any two canonical paths in $G$ that use the same inner graph $G_I^u$ have different canonical vectors, and so they use different canonical subpaths through $G_I^u$, as determined by the bijection $\phi$. 
The claim then follows from edge-disjointness of canonical paths in $G_I$.

\item Let $(s, t) \in P$, and let paths $\pi^{s_O, t_O}$ and $\pi^{s_I, t_I}$ be the canonical paths in $G_O$ and $G_I$ respectively used to define $\pi^{s, t}$.
By construction we have
$$\left|\pi^{s, t}\right| \ge \left|\pi^{s_O, t_O}\right|\left|\pi^{s_I, t_I}\right|.$$
Applying Lemma \ref{lem:base_graph} to  bound the lengths of canonical paths, we thus have
$$\left|\pi^{s, t}\right| = \Theta \left(\frac{x_O}{r_O}\cdot \frac{x_I}{r_I}\right).$$
Finally, since canonical paths of $G$ are edge-disjoint and $|W_O| = |P_I|$, we have
\begin{align*}
    |E| \geq |P| \cdot \Theta\left(\frac{x_O}{r_O}\cdot \frac{x_I}{r_I}\right) = \Omega\left(|W_I| \cdot |V|\right). \tag*{\qedhere}
\end{align*}
\end{enumerate}
\end{proof}

\section{Analysis Framework}
\label{sec:analysis_1}



Fix a critical pair of vertices $(s, t) \in P$ in our final graph $G$. 
Let $\pi^*$ denote the canonical path corresponding to critical pair $(s, t)$ in $G$, and let $\pi$ be any alternate $s \leadsto t$ path.
The majority of our analysis will be dedicated to proving that $\pi$ is much longer than $\pi^*$ in the case where $\pi$ takes at least one subdivided path not in $\pi^*$; specifically, we show $|\pi| - |\pi^*| = \Omega(r_O^{2/3})$ (see Lemma \ref{lem:gap_lem}).  
After proving this lemma, the rest of the analysis follows arguments similar to prior work \cite{AB17jacm, HP18, Lu19, LVWX22}.


We begin our analysis by decomposing paths  in $G$ into subpaths we call \textit{moves}. We define a partition of these moves that we call the moveset $\mathcal{M}$.

\begin{definition}[Moveset $\mathcal{M}$] Let $\pi$ be a $(u, v)$-path in $G$ from some input port $u \in S_I$ in some inner graph copy $G_I^{(1)}$ to some input port $v \in S_I$ in some inner graph copy $G_I^{(2)}$. If no internal vertex of $\pi$ is an input port, then we call $\pi$ a move. We define the following categories of moves in $G$.
\begin{itemize}
    \item \textbf{\textsc{Forward Move}.} Path $\pi$ is a forward move if it travels from $u$ to some output port $w \in T_I$ in $G_I^{(1)}$ and then takes a subdivided path $e$ from $w$ to reach input port $v$ in $G_I^{(2)}$.
    \item \textbf{\textsc{Backward Move}.} Path $\pi$ is a backward move if it takes some subdivided path $e$ incident to $u$ to reach some output port $w \in T_I$ in $G_I^{(2)}$ and then travels to input port $v$ in  $G_I^{(2)}$.
    \item \textbf{\textsc{Zigzag Move}.} Path $\pi$ is a zigzag move if it takes some subdivided path $e_1$ incident to $u$ to reach some output port $w_1 \in T_I$ in some inner graph copy $G_I^{(3)}$, then travels to some output port $w_2 \in T_I$ in $G_I^{(3)}$, and then takes a subdivided path $e_2$ incident to $w_2$ to reach vertex $v$ in $G_I^{(2)}$.
    \item \textbf{\textsc{Stationary Move}.} Path  $\pi$ is a stationary move if $G_I^{(1)} = G_I^{(2)}$, i.e. if $u$ and $v$ are input ports in the same inner graph copy. 
\end{itemize}
We define the moveset $\mathcal{M}$ to be the collection of these categories of moves, namely
\[\mathcal{M} = \{\textsc{Forward, Backward, Zigzag, Stationary}\}.\]
\end{definition}


\begin{figure}[h]
    \centering
\begin{tikzpicture}
\draw [black, fill=white] (0, 0) circle [radius=1];

\draw [black, fill=black] (0.6, 0.6) circle [radius=0.15];
\draw [black, fill=white] (4, 0) circle [radius=1];
\draw [black, fill=black] (3.4, 0.6) circle [radius=0.15];
\draw [red, snake it, ultra thick] (-0.6, 0.6) -- (0.6, 0.6);
\draw [red, ultra thick, ->] (0.6, 0.6) -- (3.4, 0.6);
\draw [fill=gray] (1.2, 0.6) circle [radius=0.1];
\draw [fill=gray] (1.6, 0.6) circle [radius=0.1];
\draw [fill=gray] (2, 0.6) circle [radius=0.1];
\draw [fill=gray] (2.4, 0.6) circle [radius=0.1];
\draw [fill=gray] (2.8, 0.6) circle [radius=0.1];
\node [red] at (2, 1) {forward};

\draw [black, fill=white] (-4, 0) circle [radius=1];
\draw [black, fill=black] (-3.4, 0.6) circle [radius=0.15];
\draw [black, fill=black] (-4.6, 0.6) circle [radius=0.15];
\draw [yellow!85!black, ultra thick] (-0.6, 0.6) -- (-3.4, 0.6);
\draw [yellow!85!black, snake it, ultra thick, ->] (-3.4, 0.6) -- (-4.6, 0.6);
\node [yellow!85!black] at (-2, 1) {backward};
\draw [fill=gray] (-1.2, 0.6) circle [radius=0.1];
\draw [fill=gray] (-1.6, 0.6) circle [radius=0.1];
\draw [fill=gray] (-2, 0.6) circle [radius=0.1];
\draw [fill=gray] (-2.4, 0.6) circle [radius=0.1];
\draw [fill=gray] (-2.8, 0.6) circle [radius=0.1];


\draw [black, fill=white] (-4, 4) circle [radius=1];
\draw [black, fill=black] (-3.4, 4.6) circle [radius=0.15];
\draw [black, fill=white] (0, 4) circle [radius=1];
\draw [black, fill=black] (-0.6, 4.6) circle [radius=0.15];
\draw [ultra thick, green] (-0.6, 0.6) -- (-3.4, 4.6) -- (-0.6, 4.6);

\draw [fill=gray] (-1.2, 1.5) circle [radius=0.1];
\draw [fill=gray] (-1.6, 2.05) circle [radius=0.1];
\draw [fill=gray] (-2, 2.6) circle [radius=0.1];
\draw [fill=gray] (-2.4, 3.15) circle [radius=0.1];
\draw [fill=gray] (-2.8, 3.7) circle [radius=0.1];

\draw [fill=gray] (-1.2, 4.6) circle [radius=0.1];
\draw [fill=gray] (-1.6, 4.6) circle [radius=0.1];
\draw [fill=gray] (-2, 4.6) circle [radius=0.1];
\draw [fill=gray] (-2.4, 4.6) circle [radius=0.1];
\draw [fill=gray] (-2.8, 4.6) circle [radius=0.1];

\node [green!85!black] at (-2, 4) {zigzag};

\draw [black, fill=black] (-0.6, -0.6) circle [radius=0.15];
\draw [ultra thick, snake it, blue, ->] (-0.6, 0.6) -- (0, 0) -- (-0.6, -0.6);
\draw [blue] (0, -1.4) -- (-0.3, -0.6);
\node [blue, fill=white] at (0, -1.4) {stationary};

\draw [black, fill=black] (-0.6, 0.6) circle [radius=0.3];
\node [white] at (-0.6, 0.6) {$u$};


\end{tikzpicture}
    \caption{The four types of moves from one input port to another in our move decomposition of paths.  All four moves pictured here use $u$ as their start node, the circles represent different copies of the inner graph, and the dotted lines between inner graph copies  represent subdivided paths.}
    \label{fig:my_label}
\end{figure}
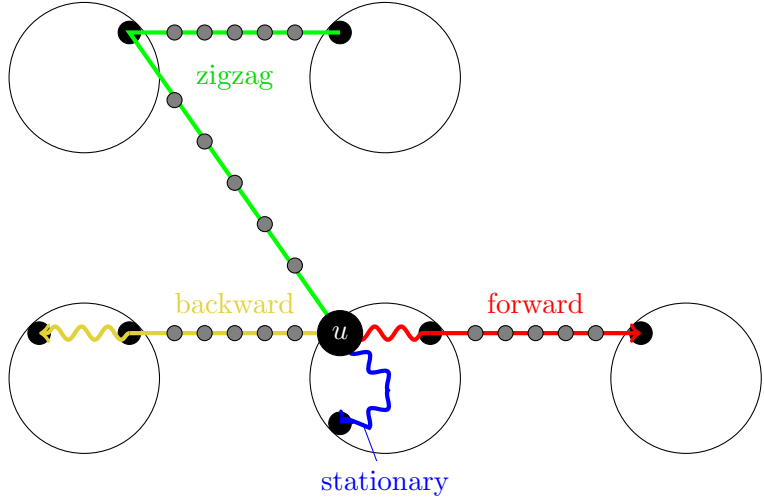

Moves will be the basic unit by which we analyze $(s, t)$-paths in $G$. A useful property of the moveset is the following.

\begin{proposition}
Every simple $(s, t)$-path $\pi$ can be decomposed into a sequence of pairwise internally vertex-disjoint moves from the moveset.
\end{proposition}
\begin{proof}
Let $s_1, s_2, \dots, s_k$ be the list of input ports contained in $\pi$, listed in their order in $\pi$. Note that $s_1 = s$ and $s_k = t$.
Each subpath $\pi[s_i, \dots, s_{i+1}]$ will have no input port as an internal vertex, and therefore will be a move $m_i$.
This move $m_i$ will be internally vertex-disjoint from all other moves $m_j$, where $i \neq j$, since $\pi$ is a simple path. It is immediate from the construction of $G$ that $\mathcal{M}$ is a partition of the set of moves in $G$. Then move $m_i$ will belong to some category of moves in $\mathcal{M}$. 
\end{proof}

Note that the canonical $(s, t)$-path $\pi^*$ specifically decomposes into a sequence of forward moves, each of which take a subdivided path corresponding to the canonical vector of $\pi^*$.
Our goal is to compare the length of $\pi^*$ to the length of an arbitrary $(s, t)$-path $\pi$. We will accomplish this by comparing the moves in the move decomposition of the two paths. We now identify some geometric notions corresponding to moves that will be useful in our analysis.

Fix a simple $(s, t)$-path $\pi$, and let $m_1, m_2, \dots, m_k$ be its move decomposition, where $m_i$ is a move from an input port $s_i$ in inner graph copy $G_I^{(i)}$ to an  input port $s_{i+1}$ in inner graph copy $G_I^{(i+1)}$, and $s_1 = s$ and $s_{k+1} = t$. If an inner graph copy $G_I^{(i)}$ replaces a vertex $v \in V_O$ in $G$, then we will let $\texttt{coord}(G_I^{(i)})$ be the vector in $\mathbb{Z}^2$ with the coordinates of $v$. We now define several geometric notions that will be essential to our analysis of path $\pi$  in $G$. 

\begin{definition}[Move vector]
The move vector $\vec{m_i}$ corresponding to  move $m_i$ is defined as
$$\vec{m_i} = \coord\left(G_I^{(i+1)} \right) - \coord\left(G_I^{(i)} \right).$$
\end{definition}

We may think of moves in the move decomposition as vectors in $\mathbb{Z}^2$ between vertices in $G_O$ (with vertices in $G_O$ corresponding to  inner graph copies $G_I$ in $G$). Now let  $\vec{v}^* \in W_O$ be the canonical vector corresponding to canonical $(s, t)$-path $\pi^*$. Corresponding to each move $m_i$, we define a move distance $d_i$.

\begin{definition}[Move distance]
The move distance $d_i$ corresponding to move vector $\vec{m}_i$ is defined as $d_i = \proj_{\vec{v}^*}\vec{m}_i$, that is, the (possibly negative) scalar projection of the vector $\vec{m}_i$ onto $\vec{v}^*$ in the standard Euclidean inner product.
\end{definition}
We roughly use $d_i$ as a measure of how much closer or farther we get to $t$ in $G$ when we take move $m_i$.
Besides the move distance $d_i$ of $m_i$, the other salient property is its length (number of edges) in the final graph, $|m_i|$.
We will be comparing the moves of a path $\pi$ against the moves of $\pi^*$, all of which have move distance $\|\vec{v}^*\|$ and the same path length. The following quantity will be useful for this purpose.


 \begin{definition}[Unit length of $\pi^*$]
 We define the unit length $L_{\pi^*}$ of $\pi^*$  as $L_{\pi^*} := \frac{|\pi^*|}{\|t - s\|}$. 
 \end{definition}

$L_{\pi^*}$ is the number of edges in $\pi^*$ per unit distance travelled in $\mathbb{Z}^2$. Using this quantity we can directly compare any move $m_i$ to the moves of $\pi^*$ via the following quantity.

\begin{definition}[Move length difference]
$\Delta(m_i) = |m_i| - L_{\pi^*}  d_i$
\end{definition}
The move length difference $\Delta(m_i)$ is the number of additional edges used by $m_i$ to travel distance $d_i$ in the direction $\vec{v}^*$, as compared with the same move in $\pi^*$.

\begin{proposition}
$\sum_i \Delta(m_i) = |\pi| - |\pi^*|$
\label{prop:move_length_sum}
\end{proposition}
\begin{proof} We have:
\begin{gather*}
\sum_i \Delta(m_i) = \sum_i |m_i| -  L_{\pi^*} \sum_i  d_i  = |\pi| -   \frac{|\pi^*|}{\|t - s\|} \cdot \sum_i d_i  =  |\pi| - |\pi^*|
\end{gather*}
The final equality follows from the fact that $\sum_i d_i = \sum_i \proj_{\vec{v}^*} \vec{m}_i = \proj_{\vec{v}^*}  (t - s)  =  \|t - s\|$, since $\pi$ is an $(s, t)$-path and $\vec{v}^*$ is the canonical vector of $(s, t) \in P$. 
\end{proof}

Proposition \ref{prop:move_length_sum} gives us a way to compare $|\pi|$ and $|\pi^*|$ at the level of individual moves. If we could show that for all moves $m$, $\Delta(m) \geq 0$, then we would be a lot closer to our current goal of proving a separation between $|\pi|$ and $|\pi^*|$. (Roughly speaking, the inequality $\Delta(m) \geq 0$ was immediate in prior constructions.) 
Unfortunately, this is not generally true in our construction.
Because our inner graph is unlayered, it's possible that the canonical inner graph path used by $\pi^*$ in  copies of inner graph $G_I$ is much longer than a different path connecting some input port to some output port in $G_I$, which might be used in an alternate move $m_i$.
This would result in negative $\Delta(m_i)$. 

We will outline our fix here, although some technical details are pushed to later in the argument where they are used.
We will use an amortized version of move difference, based on a potential function $\Phi : S_I \mapsto \rr_{\ge 0}$ that we call the \emph{inner graph potential}.
Note that the input to $\Phi$ is an input port in the original inner graph; we will evaluate $\Phi$ on input ports in various inner graph copies in the final graph, and so if $a, b$ represent the same input port in two different inner graph copies, we must have $\Phi(a)=\Phi(b)$.
We will specify $\Phi$ in Section \ref{sec:analysis_2}.
With this, we can define an amortized version of move difference, $\widehat{\Delta}(m_i)$.

\begin{definition}[Amortized move difference]
$\widehat{\Delta}(m_i) := \Delta(m_i) - \left(\Phi\left(s_{i+1}\right) - \Phi\left(s_i\right)\right)$ 
\end{definition}

The following proposition shows that the amortized move difference still captures the difference between $|\pi|$ and $|\pi^*|$. 

\begin{proposition}
$\sum_i\widehat{\Delta}(m_i) = |\pi| - |\pi^*| $ 
\label{prop:amort_diff}
\end{proposition}
\begin{proof}
We have:
$$\sum_i\widehat{\Delta}(m_i) =  \sum_i \Delta(m_i) - \sum_i(\Phi(s_{i+1}) - \Phi(s_i)) = |\pi| - |\pi^*| - (\Phi(t) - \Phi(s)) = |\pi| - |\pi^*|.$$
The final equality follows from the fact that  $s$ and $t$ have the same coordinates in their respective inner graph copies as specified in the definition of $P$. 
\end{proof}

In Section \ref{sec:analysis_2}, we will see that  $\widehat{\Delta}(m) \geq 0$ for every move $m$. Then using Proposition \ref{prop:amort_diff} and the unique shortest path property of the base graph $G_B$ (see Lemma \ref{lem:base_graph}), we will obtain our desired separation between $|\pi|$ and $|\pi^*|$.

\section{Specifying the Construction \label{sec:specifying}}

In the following notation, we will let $n_I = |V_I|$ and $n_O = |V_O|$.
We also use a parameter $c > 0$; roughly one may think of $c$ as a large constant, and we will assume where convenient that $c$ is \emph{at least} a large enough constant.
But,  we do not hide $f(c)$ factors in our big-O notation unless explicitly noted with $O_c$ notation.
Our goal is to argue that  for our final graph $G$ and for any subgraph $H \subseteq G$ with $\leq cn$ edges, there exists a pair of vertices $u, v \in V(G)$ such that
$$d_H(u, v) > d_G(u, v) + \Omega_c\left(n^{1/7}\right).$$
This is a rephrasing of the statement that $O(n)$-size spanners generally require $\Omega(n^{1/7})$ error.
We may also assume where convenient that $n_I, n_O$ are sufficiently large, relative to $c$.
We do not make any effort to optimize the dependence of our lower bound on $c$; to do so would imply stronger lower bounds against denser spanners, but it introduces considerable technical complexity that we do not think is worth it.
See \cite{BV16} for discussion of these optimizations.

\subsection{Specifying the Inner Graph}


We next specify the parameters $x_I, y_I, r_I, W_I$ used to construct the inner graph $G_I = (V_I, E_I)$.
Relative to a choice of $n_I$ for the \emph{total} number of nodes in the inner graph, we use dimensions
$$x_I = 2^{-1}n_I^{1/2}\cdot c^{50} \qquad \text{and} \qquad y_I = 2n_I^{1/2} \cdot c^{-50}$$
and so $x_I y_I = n_I$. 
We set $r_I = c^{102}$; note that by choice of large enough $n_I$ we have $r_I \leq x_I/4$.
We then define $W_I$ to be the set of $c$ vectors
$$\left\{\left(r_I - c + i, \sum_{j=i}^c j\right) \mid i \in [1, c] \right\}.$$
Strong convexity of $W_I$ follows from the fact that the function $f:\mathbb{Z} \mapsto \mathbb{Z}$ defined as $f(x) = \sum_{j=x}^c j$ is positive and strictly concave on the interval $[1, c]$.
We also notice that for all $\vec{v} \in W_I$, the first coordinate of $\vec{v}$ is in the range $[r_I/2, r_I]$. 
Let $\psi_I$ be the largest angle between a vector in $W_I$ and the horizontal.
Observe that by taking $c$ to be sufficiently large, we have
$$\tan \psi_I \leq \left(r_I - c\right)^{-1} \cdot \sum_{j=1}^cj \leq 2c^{-100} \le \frac{y_I x_I^{-1}}{2}.$$
Thus this inner graph construction satisfies the premises of the base graph, and we have proved:
\begin{lemma}[Inner Graph Strongly Convex Set]
\label{lemma:inner_graph_SCS}
The set $W_I(r_I, \psi_I)$ has the following properties:
\begin{enumerate}
    \item $|W_I| = c$
    \item For all $\vec{v} = (v_1, v_2) \in W_I$, $r_I/2 \leq v_1 \leq r_I$. Specifically, $W_I \subseteq [r_I - c, r_I] \times [0, c^2]$
    \item $0 \leq \tan \psi_I \leq 2c^{-100}$. 
\end{enumerate}
\end{lemma}

The next lemma is a key structural property of our inner graph that enables its analysis in our move decomposition framework.
To enable this lemma, we need to specifically choose $n_I$ such that, for each canonical vector $(v_1, v_2) \in W_I$, we have $v_1 \mid (x_I - r_I/2)$.
In particular, this can be accomplished by defining $\lambda := \Pi_{i=1}^c (r_I - c + i)$ and choosing $n_I$ so that $x_I \equiv r_I/2 \mod \lambda$. Since $\lambda = \Theta_c(1)$, there are infinitely many choices of $n_I$ that satisfy this.
\begin{proposition}
For all $((s_1, s_2), (t_1, t_2)) \in P_I$, we have $t_1 - s_1 = x_I - r_I/2$. 
\label{prop:P_I_x_diff}
\end{proposition}
\begin{proof}
Let $\vec{v} = (v_1, v_2)$ be the canonical vector for $((s_1, s_2), (t_1, t_2)) \in P_I$.
By choice of $n_I$ we have that $v_1 \mid (x_I - r_I/2)$.
Now let $k = (x_I - r_I/2) / v_1$, and let $t' = (s_1, s_2) + k\vec{v}$. Observe that
$$t' = (s_1, s_2) + k \vec{v} = (s_1 + kv_1, s_2 + kv_2) = \left(s_1 + x_I - r_I/2,  s_2 + v_2/v_1 \cdot (x_I - r_I/2)\right).$$ 
Since $s \in S_I$ we have $s_1 \in [1, r_I/2]$, and so $t'$ is a vertex in $[x_I - r_I/2 + 1, x_I] \times [1, y_I] \subseteq T_I$.
We then note that $v_1 \geq r_I/2$ by property 2 of Lemma \ref{lemma:inner_graph_SCS}.
So the first coordinate of $t' + \vec{v}$ is $>s_1 + x_I$, and thus $t' + \vec{v} \not \in T_I$.
It follows that $t' = (t_1, t_2)$, and so $t_1 = s_1 + x_I - r_I/2$ as desired.
\end{proof}

 Proposition \ref{prop:P_I_x_diff} guarantees that all critical pairs $(s, t) \in P_I$ have the same horizontal displacement.
Note that this doesn't imply that the \emph{graph} distances between critical pairs in $G_I$ is the same. 
Applying Lemma \ref{lem:base_graph} to our construction, we summarize the following properties of the inner graph:

\begin{lemma}[Inner Graph Properties]
Inner graph $G_I = (V_I, E_I)$ has the following properties:
\begin{enumerate}
    \item $|V_I| = n_I$
    \item $|E_I| = \Theta(c \cdot n_I)$
    \item $|P_I| = \Theta( c^{53} \cdot n_I^{1/2})$
    \item The canonical paths $\pi_I^{s, t}$ for $(s, t) \in P_I$ are pairwise edge-disjoint.
    \item The canonical path $\pi_I^{s, t}$ is the unique shortest $(s, t)$-path in $G_I$ for all $(s, t) \in P_I$. 
\end{enumerate}
\end{lemma}

\subsection{Specifying the Outer Graph}

In our specification of the the outer graph $G_O$, we will make use of the existence of dense strongly convex sets of integer vectors in $\mathbb{Z}^2$ with certain nice properties. The existence of these sets follows directly from the work of \cite{BL98} on the size of the convex hull of integer points inside a disk of radius $r$; however, we need to tweak the construction slightly to enforce a few extra convenient properties.
Most of the following lemma follows directly from \cite{BL98}, but for completeness we give a proof in Appendix \ref{app:CIS}.

\begin{lemma}
\label{lem:CIS}
For sufficiently large $r$, there exists a strongly convex set $W(r)$ of integer vectors in $\mathbb{Z}^2$  of size $\Theta(r^{2/3})$ such that:
\begin{enumerate}
    \item For all $\vec{v} \in W(r)$, $r - r^{-1/3} \leq \|\vec{v}\| \leq r$. 
    \item If $S$ is a  sector with inner angle $\psi$
    of the circle of radius $r$ centered at the origin, then there are at most $O(\psi \cdot r^{2/3})$ vectors in $W(r) \cap S$.
    \item For all distinct $\vec{u}, \vec{v} \in W(r)$, $\proj_{\vec{v}}\vec{u}  < \|\vec{v}\|$.
    (In other words, for all vectors $\vec{v} \in W(r)$, the vector in $W(r)$ with the largest magnitude scalar projection in the direction of $\vec{v}$ is $\vec{v}$ itself.  This implies strong convexity, but is in fact a bit stronger.)
\end{enumerate}
\end{lemma}
\begin{proof}
Deferred to Appendix \ref{app:CIS}.
\end{proof}

Now we use the convex set $W(r)$ constructed in Lemma \ref{lem:CIS} as a starting point for our construction of $W_O(r_O, \psi_O)$.

\paragraph{Strongly convex set $W_O(r_O, \psi_O)$ for $G_O$.}  
Let $W(r_O)$ denote the strongly convex set of Lemma \ref{lem:CIS} with input parameter $r_O$. By Lemma \ref{lem:CIS},  there exists a circular sector  with inner angle $\psi_1 = c^{-5}$ radians that contains $\Omega\left(\psi_1 r_O^{2/3}\right)$ vectors from $W(r_O)$. By choice of sufficiently large $c$, we may assume $\psi_1 \leq \pi / 4$.  We let $W'(r_O)$ be the set of vectors in $W(r_O)$ with endpoints in this sector. 

We will take parameter $\psi_O = \pi/4$, i.e. we will modify $W'(r_O)$ so that its vectors lie in the first quadrant and have angle at most $\psi_O$ with the horizontal. Note that if we reflect the vectors of $W(r_O)$ across the lines $x = 0$, $y=0$,  $y=x$, or $y = -x$ in $\mathbb{R}^2$, then the resulting set of vectors $W'(r_O)$ also satisfies the properties of Lemma \ref{lem:CIS}. By performing these reflection operations a constant number of times, we can ensure that at least half the vectors in $W'(r_O)$ lie in the first quadrant and have angle at most $\psi_O$ with the horizontal.  We let $W''(r_O)$ denote the resulting set of vectors with maximum angle $\psi_O$ to the horizontal; we thus have $|W''(r_O)| = \Omega(c^{-5}r_O^{2/3})$.

We are now ready to construct our strongly convex set $W_O$ from $W''(r_O)$. 
Our set $W_O$ will be partitioned into $c$ disjoint sets $\mathcal{S}_1, \mathcal{S}_2, \dots, \mathcal{S}_c$ called \textit{stripes}, which have the following two properties:
\begin{itemize}
    \item (Same Size) All stripes contain the same number of vectors $|\mathcal{S}_i| =: \beta = \Theta(c^{-6} r_O^{2/3})$ vectors. 
    \item (Well-Separated) For any two distinct stripes $\mathcal{S}_i, \mathcal{S}_j$ and for vectors $\vec{u} \in \mathcal{S}_i, \vec{v} \in \mathcal{S}_j$, the angle between $\vec{u}, \vec{v}$ is at least $c^{-10}$. 
\end{itemize}
We construct our stripes as follows.
Starting at the horizontal and rotating counterclockwise, let $\vec{u}_1^1, \vec{u}_2^1, \dots, \vec{u}_{\beta}^1$ be  the first $\beta$ vectors we encounter in $W''(r_O)$. 
After encountering the $\beta$th vector, we rotate $c^{-10}$ radians about the origin counterclockwise, ignoring any vectors encountered in this arc.
By Lemma \ref{lem:CIS}, this will skip over only $O(c^{-10} r_O^{2/3})$ vectors in $W''(r_O)$.
We then take the next $\beta$ vectors we encounter to be  $\vec{u}_1^2, \vec{u}_2^2, \dots, \vec{u}_{\beta}^2$, and again rotate $c^{-10}$ radians counterclockwise.  We repeat this process until we've obtained vectors $\vec{u}_i^j$ for $i \in [1, \beta]$ and $j \in [1, c]$. 
 We call the set of of vectors $\{ \vec{u}_i^j \mid i \in [1, \beta]\}$ the $j$th stripe $\mathcal{S}_j$ of $W_O$.
Our procedure is guaranteed to identify all $c\beta$ vectors $\vec{u}_i^j$, since we skip over only $O(c^{-9}  r_O^{2/3})$ vectors from $W''(r_O)$ in our construction.
Note that $|W_O| = c\beta = \Theta(r_O^{2/3}c^{-5})$.  
We summarize our properties:

\begin{lemma} [Outer Graph Strongly Convex Set]
\label{lem:outer_graph_SCS}
For sufficiently large parameters $c, r_O$ we may construct a strongly convex set $W_O = W_O(r_O, \psi_O)$ containing $\Theta(r_O^{2/3}c^{-5})$ integer vectors in $\mathbb{Z}^2$ with the following properties:
\begin{enumerate}
    \item For all $\vec{v} \in W_O$, $r_O - r_O^{-1/3} \leq \|\vec{v}\| \leq r_O$. 
    \item For all $\vec{u}, \vec{v} \in W_O$, $\proj_{\vec{v}}\vec{u}  \leq \|\vec{v}\|$.
    \item For all $\vec{v} \in W_O$, the angle between vector $\vec{v}$ and the horizontal is at most $\psi_O = \pi/4$  radians.
    \item For all $\vec{u}, \vec{v} \in W_O$, the angle between $\vec{u}$ and $\vec{v}$ is at most $\psi_1 = c^{-5}$ radians.
    \item $W_O$ can be partitioned into $c$ stripes $\mathcal{S}_1, \mathcal{S}_2, \dots, \mathcal{S}_c$ with the following properties:
    \begin{enumerate}
        \item $|\mathcal{S}_i| = \beta =  \Theta(r_O^{2/3}c^{-6})$ for all $i \in [1, c]$. 
        \item For $\vec{u} \in \mathcal{S}_i$ and $\vec{v} \in \mathcal{S}_j$ where $i \neq j$, the angle between $\vec{u}$ and $\vec{v}$ is at least $\psi_2 = c^{-10}$ radians. 
    \end{enumerate}
\end{enumerate}
\end{lemma}


\begin{figure}[h]
    \centering   \includegraphics[scale=0.5 ]{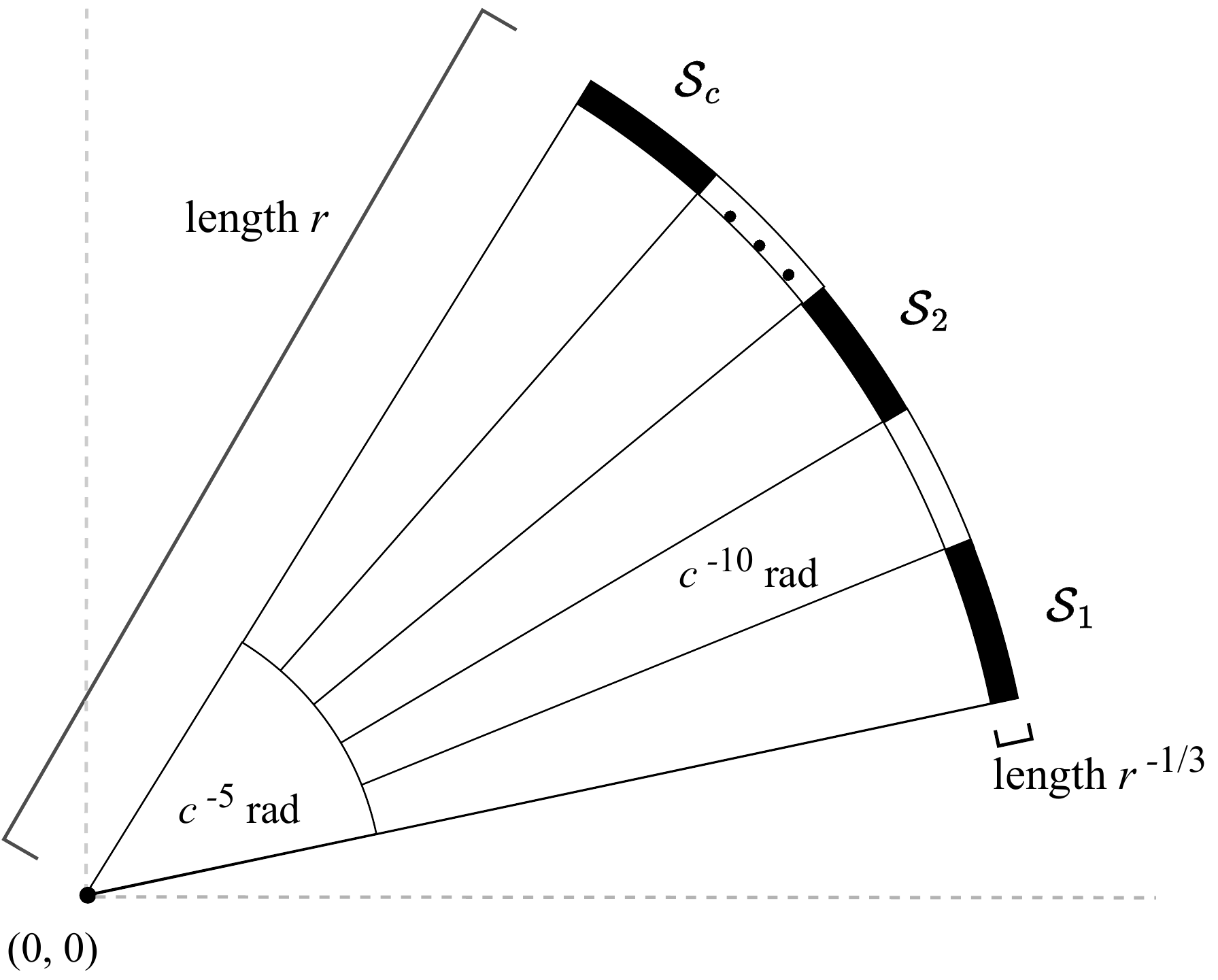}
    \caption{A depiction of $W_O$. The endpoints of the vectors in $W_O$ lie in the black stripes in the above figure.}
    \label{fig:stripes}
\end{figure}

\paragraph{Parameter choices for $G_O$.} We will let $2x_O = y_O = \sqrt{2}n_O^{1/2}$.  We will leave $r_O$ unspecified until the end of the construction, for expository purposes, as it is selected according to a parameter balance on the final graph. However, unlike with $r_I$ in $G_I$, our choice of $r_O$ will grow polynomially with $n_O$.
We quickly verify that our choices satisfy the base graph construction.
We have $\tan \psi_O = 1 = y_Ox_O^{-1}/2$.
Additionally, for all $\vec{v} \in W_O$, the first coordinate $v_1$ of $\vec{v}$ is between $r_O/2$ and $r_O$, by properties 1 and 3 of Lemma  \ref{lem:outer_graph_SCS}.
Then assuming we choose $r_O \leq n_O^{1/2}/4$, it's clear that our parameter choice $(x_O, y_O, r_O, W_O)$ of $G_O$ will yield a valid instantiation of base graph $G_B$.
The following lemma summarizes the properties of $G_O$, using Lemma \ref{lem:base_graph}:

\begin{lemma}[Outer Graph Properties]
\label{lemma:outer_graph}
Outer graph $G_O = (V_O, E_O)$ has the following properties:
\begin{enumerate}
    \item $|V_O| = n_O$
    \item $|E_O| = \Theta(c^{-5}\cdot r_O^{2/3}\cdot n_O)$
    \item $|P_O| = \Theta(c^{-5} \cdot r_O^{5/3} \cdot n_O^{1/2})$
    \item The canonical paths $\pi_O^{s, t}$ for $(s, t) \in P_O$ are pairwise edge-disjoint.
    \item The canonical path $\pi_O^{s, t}$ is the unique shortest $(s, t)$-path in $G_O$ for all $(s, t) \in P_O$. 
\end{enumerate}
\end{lemma}

\subsection{Specifying the Final Graph \label{sec:specifyingfinal}}

In our construction of final graph $G$, we made use of a bijection $\phi: W_O \mapsto P_I$ to plug the subdivided paths of $G_O$ into the input and output ports of the inner graph copies of $G_I$. To that end, we required that $|W_O| = |P_I|$. This can be achieved by taking $n_I = \Theta(r_O^{4/3}c^{-116})$, so our inner graph will be fully specified by parameters $c, r_O$. Note that we can get exact equality in the cardinality of $W_O$ and $P_I$ by simply ignoring a constant fraction of vectors in $W_O$ or a constant fraction of paths in $P_I$. 

We specifically define $\phi$ as follows. 
Let $\vec{v}_1, \vec{v}_2, \dots, \vec{v}_c$ be the $c$ vectors in $W_I$.  Observe that for every $\vec{v}_i \in W_I$, each vertex $s \in S_I$ has a critical pair in $P_I$ with $\vec{v}_i$ as its canonical vector. Then there will be exactly $|P_I|c^{-1}$ critical pairs in $P_I$ with canonical vector  $\vec{v}_i$. Let $\phi$ map the  $|W_O|c^{-1}$ vectors in stripe $\mathcal{S}_i$ of $W_O$ to the  $|P_I|c^{-1}$ critical pairs in $P_I$ with canonical vector $\vec{v}_i$, with some arbitrary bijection, for all $i \in [1, c]$.
The following proposition captures the key property of $\phi$ needed in our construction.
\begin{proposition}
 For all $\vec{u}, \vec{v} \in W_O$, the critical pairs $\phi(\vec{u}), \phi(\vec{v}) \in P_I$ have the same canonical inner graph vector in $W_I$ if and only if $\vec{u}$ and $\vec{v}$ belong to the  same stripe $\mathcal{S}_i$, $i \in [1, c]$.
 \label{prop:phi}
\end{proposition}

For analysis purposes, it will be convenient if the set of canonical paths for pairs in $P$ partition the edges $E$ of $G$.
Currently, the canonical paths partition the \emph{outer} extended edges of $G$, since in our base graph construction we only include edges that lie in canonical paths.
However, it might be that some \emph{inner} graphs edges $e$ are not used in canonical paths for pairs in $P$.
Specifically, this happens in the case where the inner-canonical path containing $e$ bijects with a canonical vector $\vec{v}$ that, in turn, does not define an edge contained in an outer canonical path (typically because the point corresponding to the inner graph that contains $e$, plus the vector $\vec{v}$, yields a point outside the $[1, x] \times [1, y]$ dimensions of the outer graph).

To simplify the following analysis we remove all edges in $E$ from inner graphs that do not belong to a canonical path $\pi^{s, t}$ where $(s, t) \in P$. Then by properties 3 and 4 of Lemma \ref{lem:final_graph}, the canonical paths for pairs in $P$ will partition $E$ and the density of $G$ will remain the same (within constant factors).
We summarize the properties of our final graph:

\begin{lemma}[Properties of Final Graph $G$]
\label{lem:final_graph_2}
Graph $G = (V, E)$ has the following properties:
\begin{enumerate}
    \item $|V| = \Theta( r_O^{4/3} \cdot c^{-116} \cdot n_O )$
    \item $|E| = \sum_{(s, t) \in P} |\pi^{s, t}| =  \Theta(c \cdot |V|)$ 
    \item $|P| = \Theta(r_O^{5/3} \cdot c^{-5} \cdot n_O^{1/2})$
    \item The subdivided paths of $G$ are of length $z = |V_I| |P_I|^{-1} = \Theta(r_O^{2/3} c^{-111})$.
    \item The set of canonical paths $\{\pi^{s, t} \mid (s, t) \in P \}$ form a partition of $E$. 
\end{enumerate}
\end{lemma}

\section{Completing the Analysis}
\label{sec:analysis_2}
As in Section \ref{sec:analysis_1}, fix a critical pair of vertices $(s, t) \in P$ in our final graph $G$ with canonical path $\pi^*$ and canonical vector $\vec{v}^*$. We will compare $\pi^*$ to an arbitrary simple $(s, t)$-path $\pi$, with move decomposition $m_1, m_2, \dots, m_k$. As before, move $m_i$ is a move from an input port $s_i$ in inner graph copy $G_I^{(i)}$ to an input port $s_{i+1}$ in inner graph copy $G_I^{(i+1)}$, and $s_1 = s$ and $s_{k+1} = t$. 

 The majority of the analysis in this section will be towards giving lower bounds for the amortized move difference $\widehat{\Delta}(m_i)$ of moves $m_i$. Intuitively, we should think of proving lower bounds on  $\widehat{\Delta}(m_i)$ as proving that move $m_i$ is longer than a  move in the canonical path in some sense. 
In this section, it will be helpful to recall that $r_I = c^{102}$, $x_I = 2^{-1}n_I^{1/2} c^{50}$, $z= \Theta(n_I^{1/2}c^{-53})$, and $n_I = \Theta(r_O^{4/3} c^{-116})$. 

\subsection{Lower Bounding the Move Difference $\Delta$}
We begin our analysis by lower bounding the move difference $\Delta$ of several categories of moves in the moveset. We will need the following useful proposition.

\begin{proposition}[Inner Graph Path Lengths] \leavevmode
\begin{enumerate}
    \item For every $(s_I, t_I)$-path $\pi_I$ in $G_I$, where $s_I \in S_I$ and $t_I \in T_I$,
\[
\frac{n_I^{1/2}}{2c^{52}} - 2  \leq |\pi_I| 
\]
    \item For every critical pair $(s_I, t_I) \in P_I$ with canonical path $\pi_I^{s_I, t_I}$, 
    \[
    |\pi_I^{s_I, t_I}| \leq \frac{n_I^{1/2}c^{50}}{2(c^{102} - c)}
    \]
\end{enumerate}
\label{prop:inner_path_lengths}

\end{proposition}
\begin{proof} ~
\begin{enumerate}
    \item  Observe that by the construction of $G_I$, the $x$-displacement between any vertex in $S_I$ and any vertex in $T_I$ is at least $x_I - 2r_I = 2^{-1}n_I^{1/2} c^{50} - 2c^{102}$. Additionally, each edge $e \in E_I$ corresponds to a vector $\vec{w} = (w_1, w_2) \in W_I$, and by property 2 of Lemma \ref{lemma:inner_graph_SCS}, $w_1 \leq r_I =  c^{102}$. Then for every $s \in S_I$ and $t \in T_I$, every $(s_I, t_I)$-path in $G_I$ is of length at least 
\[
\frac{ 2^{-1}n_I^{1/2} c^{50} - 2c^{102}}{c^{102}} = \frac{n_I^{1/2}}{2c^{52}} - 2
\]
\item The $x$-displacement between any vertex in $S_I$ and any vertex in $T_I$ is less than $x_I = 2^{-1}n_I^{1/2} c^{50}$.  Let $\vec{v}_I^* = (v_{I, 1}^*, v_{I, 2}^*)$ be the canonical vector in $W_I$ corresponding to $\pi_I^{s_I, t_I}$, and note that by property 2 of Lemma \ref{lemma:inner_graph_SCS}, $v_{I, 1}^* \geq r_I - c = c^{102} - c$. Then $|\pi_I^{s_I, t_I}|$ is a most $(2^{-1}n_I^{1/2} c^{50}) / ( c^{102} - c)$ as desired. 
\end{enumerate} 
\end{proof}

Now we will lower bound the move length difference $\Delta(m_i)$ of moves in \textsc{Backward}. This lemma essentially formalizes the obvious reason why backwards moves are much worse than canonical moves -- it's because they move backwards, away from $t$. 

\begin{lemma}
Let $m_i$ be a move in \normalfont \textsc{Backward}.
\textit{Then}
$\Delta(m_i) = \Omega(r_O^{2/3}c^{-110})$.
\label{lem:backward}
\end{lemma}
\begin{proof}
Observe that in a backward move, we take a subdivided path corresponding to some vector $\vec{m}_i = -\vec{u}$ for some $\vec{u} \in W_O$. Since all vectors in $W_O$ lie in a sector with inner angle $c^{-5}$ by property 3 of Lemma \ref{lem:outer_graph_SCS}, it follows that the angle between $\vec{u}$ and $\vec{v}^*$ is less than $c^{-5}$, so the angle between $\vec{m}_i = -\vec{u}$ and $\vec{v}^*$ is at least $\pi - c^{-5}$. 
Then for sufficiently large $c, r_O$, the scalar projection of $\vec{m}_i$ onto $\vec{v}^*$ will be negative, so $d_i < 0$.
Recall that $L_{\pi^*}$ is the unit length of $\pi^*$ and that $L_{\pi^*} \geq 0$. Then 
\[
\Delta(m_i) = |m_i| - L_{\pi^*}  d_i  \geq |m_i| \geq \frac{n_I^{1/2}}{2c^{52}} - 2 + z  = \Omega(r_O^{2/3}c^{-110})
\]
(The final inequality follows from part 1 of Proposition \ref{prop:inner_path_lengths}.)
\end{proof}

We now give a precise upper bound on $L_{\pi^*}$, the unit length of $\pi^*$.

\begin{proposition} $L_{\pi^*} = O(r_O^{-1/3}c^{-110})$, and specifically,
\[
L_{\pi^*} \leq \frac{n_I^{1/2}c^{50} \cdot(2(c^{102} - c))^{-1} + z}{r_O - r_O^{-1/3}}
\]
\label{prop:path_length_factor}
\end{proposition}
\begin{proof}
Recall that   $z$ is the length of all subdivided paths in $G$. 
Note that every move $m$ in path $\pi^*$ is a forward move with the same move vector $\vec{m} = \vec{v}^*$ and will have the same path length $|m|$. Then an upper bound for $|m| \cdot \|\vec{v}^*\|^{-1}$ will give an upper bound for $L_{\pi^*}$. 
By property 1 of Lemma \ref{lem:CIS}, $\|\vec{v}^*\| \geq r_O - r_O^{-1/3}$ for sufficiently large $r_O$.  
Note that the length of $m$ is the length of the canonical inner graph path corresponding to $\pi^*$, plus the length $z$ of the subdivided path taken by $m$. By part 2 of Proposition \ref{prop:inner_path_lengths}, every canonical path in $G_I$ is of length at most $(n_I^{1/2}c^{50})/(2(c^{102} - c))$. Then combining these observations,
\[
L_{\pi^*} = \frac{|m|}{  \|\vec{v}^*\|} \leq \frac{(n_I^{1/2}c^{50})\cdot(2(c^{102} - c))^{-1} + z}{r_O - r_O^{-1/3}}
\]
as claimed. For sufficiently large $c, r_O$, we find that 
$L_{\pi^*} \leq \frac{4n_I^{1/2}}{r_Oc^{52}} = O(r_O^{-1/3}c^{-110})$.
\end{proof}

We now lower bound the move length difference $\Delta(m_i)$ of zigzag moves.  We will see that because we specified that the vectors in our set $W_O$ lie in a small cone with inner angle $\psi_1 = c^{-5}$ (see property 4 of Lemma \ref{lem:outer_graph_SCS}), zigzag moves will be quite inefficient compared to canonical moves. 

\begin{lemma}
Let $m_i$ be a move in \normalfont \textsc{Zigzag}. Then $\Delta(m_i) = \Omega(r_O^{2/3}c^{-111})$.  \label{lem:zigzag}
\end{lemma}
\begin{proof}
Observe that in a zigzag move, we first take some vector $-\vec{u}_1$ for some $\vec{u}_1 \in W_O$, and then we take some other vector $\vec{u}_2$ with $\vec{u}_2 \in W_O$. Then $\vec{m}_i = \vec{u}_2 - \vec{u}_1$. Since $\vec{u}_1, \vec{u}_2 \in W_O$, the angle between vector $\vec{u}_1$ and $\vec{u}_2$ is at most $c^{-5}$ by property 4 of Lemma \ref{lem:outer_graph_SCS}. Then since $r_O - r_O^{-1/3} \leq \|\vec{u}_1\| \leq r_O$ and $r_O - r_O^{-1/3} \leq \|\vec{u}_2\| \leq r_O$, some straightforward geometry shows that $\|\vec{m}_i\| \leq 2r_O^{-1/3} + c^{-5} \cdot r_O$. 

Then $d_i \leq \|\vec{m}_i\| \leq 2c^{-5}\cdot r_O$ for sufficiently large $r_O$. Additionally, observe that move $m_i$ takes two subdivided paths in $G$, so $|m_i| \geq 2z = \Omega(r_O^{2/3}c^{-111})$. Then 
\[
\Delta(m_i) = |m_i| - L_{\pi^*}  d_i \geq |m_i| - L_{\pi^*}  2c^{-5}r_O \geq  2z - O(r_O^{2/3}c^{-115}) = \Omega(r_O^{2/3}c^{-111})
\] 
for sufficiently large $c$.
\end{proof}

Now to analyze the forward moves in $\mathcal{M}$, we find it useful to partition the set \textsc{Forward} into sets \textsc{Forward-S} and \textsc{Forward-D}, depending on the vector in $W_O$ corresponding to the subdivided path taken in the move. Let $m$ be a move in \textsc{Forward}. Then $m$ is in \textsc{Forward-S}  if move vector $\vec{m} \in W_O$  belongs to the \textit{same} stripe as vector $\vec{v}^*$ in $W_O$. Otherwise, move vector $\vec{m}$ belongs to a \textit{different} stripe than $\vec{v}^*$ and $m$ is in \textsc{Forward-D}. We first analyze moves in \textsc{Forward-D}.

\begin{proposition}
Let $\vec{u}, \vec{v}$ be two vectors in $W_O$ that belong to different stripes. Then 
$\proj_{\vec{v}} \vec{u} \leq r_O(1 -  \sigma)$, where $\sigma = 4^{-1}c^{-20}$. Consequently, if $m_i$ is a move in \normalfont \textsc{Forward-D}, then $d_i \leq r_O(1 - \sigma)$.
\label{prop:stripe_diff}
\end{proposition}
 \begin{proof}
 By  property 4a of Lemma \ref{lem:outer_graph_SCS}, the angle $\psi$ between $\vec{u}$ and $\vec{v}$ is at least $\psi_2 = c^{-10}$ radians. Then from the Taylor expansion of $\cos x$, we get that $\proj_{\vec{v}} \vec{u} = \|\vec{u}\| \cos \psi \leq  r_O \cdot (1 - \psi^2/2 + \psi^4/4!) \leq  r_O(1 -  \sigma)$ for sufficiently large $c$. 
 Now note that if $m_i$ is a move in \textsc{Forward-D}, then $d_i = \proj_{\vec{v}^*} \vec{m}_i \leq r_O(1 -  \sigma)$.
 \end{proof}
 
 Proposition \ref{prop:stripe_diff} tells us that if a move $m_i$ is in \textsc{Forward-D}, then it loses a constant fraction $\sigma$ of its move distance $d_i$. This follows from the fact that $\vec{m}_i$ is in a different stripe than $\vec{v}^*$, so move vector $\vec{m}_i$ is not travelling as far in the direction of $\vec{v}^*$. (See Figure \ref{fig:stripes} for reference.) 
 This deficiency in $d_i$ makes moves in \textsc{Forward-D} inefficient compared to canonical moves, as we prove in the following lemma.
 
 \begin{lemma}
 Let $m_i$ be a move in \normalfont \textsc{Forward-D}. Then $\Delta(m_i) = \Omega(r_O^{2/3}c^{-130})$. 
 \label{lem:forwardd}
 \end{lemma}
 \begin{proof}
 Since $m_i$ is a move in \textsc{Forward-D}, it follows that $d_i \leq r_O(1 - \sigma)$ by Proposition \ref{prop:stripe_diff}. Note that the length of $m_i$ is the length of the inner graph traversal of $G_I$ in $m_i$ plus the length $z$ of the subdivided path taken by $m$. Then by part 1 of Proposition \ref{prop:inner_path_lengths}, move $m_i$ requires at least  $(n_I^{1/2})/(2c^{52}) - 2$ edges to travel from an input port in $G_I$ to an output port. The length of the subdivided path taken by $m_i$ is $z$, so $|m_i| \geq (n_I^{1/2})/(2c^{52}) - 2 + z$. We calculate $\Delta(m_i)$ as follows. 
 \begin{align*}
     \Delta(m_i) & = |m_i|  - L_{\pi^*}  d_i \\
     & \geq |m_i| -   \frac{n_I^{1/2}c^{50} \cdot(2(c^{102} - c))^{-1} + z}{r_O - r_O^{-1/3}} \cdot r_O(1 - \sigma) & \text{ by Prop. \ref{prop:path_length_factor}, \ref{prop:stripe_diff}}  \\
     & \geq |m_i| -   (n_I^{1/2}c^{50} \cdot(2(c^{102} - c))^{-1} + z) \cdot (1 + 2r_O^{-4/3}) \cdot (1 - \sigma)  \\
     & \geq \left(\frac{n_I^{1/2}}{2c^{52}} - 2 + z\right) -   \left( \frac{n_I^{1/2}c^{50}}{2(c^{102} - c)}  + z\right)  \cdot (1 - \sigma) - O(r_O^{-2/3}) & \text{ by Prop. \ref{prop:inner_path_lengths}} \\
     & \geq \frac{n_I^{1/2}}{2c^{52}}  -   \frac{n_I^{1/2}c^{50}}{2(c^{102} - c)}   \cdot (1 - \sigma) - O(r_O^{-2/3})  \\
      & \geq \left( \frac{n_I^{1/2}}{2c^{52}}  -   \frac{n_I^{1/2}c^{50}}{2(c^{102} - c)} \right)  + \frac{n_I^{1/2}c^{50}}{2(c^{102} - c)}\sigma  - O(r_O^{-2/3})  \\
            & \geq -\frac{n_I^{1/2}}{c^{153}}  + \frac{n_I^{1/2}}{8c^{72}}  - O(r_O^{-2/3}) = \Omega(r_O^{2/3}c^{-130}) \tag*{\qedhere}
 \end{align*}
 \end{proof}
 

We have shown that $\Delta(m_i) = \Omega(r_O^{2/3}c^{-130})$ if $m_i$ is a move in \textsc{Backward}, \textsc{Zigzag}, or \textsc{Forward-D}.  Moreover,  if $m_i$ is a move in \textsc{Stationary}, then $\Delta(m_i) \geq 0$ since $\vec{m}_i = \vec{0}$. However, it's not true in general that $\Delta(m_i) \geq 0$ if $m_i$ is a move in \textsc{Forward-S}. This is essentially because $m_i$ can take an $(S_I, T_I)$-path in $G_I$ that is possibly much shorter than the canonical inner graph path used by $\pi^*$. In the following section we will introduce our inner graph potential function $\Phi$, and show that in an amortized sense, moves  in \textsc{Forward-S} are not shorter than the moves in $\pi^*$.

\subsection{Lower Bounding the Amortized Move Difference $\widehat{\Delta}$}

 Let $(s^*, t^*) \in P_I$ be the inner graph critical pair whose canonical path $\pi_I^{s^*, t^*}$ is a subpath of $\pi^*$ in $G_I$, and let $s^* := (s_1^*, s_2^*)$ and $t^* := (t_1^*, t_2^*)$. 
 Note that by construction of $G$,  $(s^*, t^*) = \phi(\vec{v}^*)$. 
Additionally, let  $\vec{v}_I^* := (v_{I, 1}^*, v_{I, 2}^*) = (r_I - c + i^*, \sum_{j=i^*}^c j) \in W_I$ be the canonical inner graph vector corresponding to  critical pair $(s^*, t^*) \in P_I$, where $i^* \in [1, c]$. Then for $s_i = (s_{i, 1}, s_{i, 2}) \in S_I$ we define our inner graph potential function $\Phi$  to be 
\[\Phi(s_i) = \frac{s_{i, 2} \cdot ( i^*)^{-1} + s_{i, 1}}{r_I - c + i^*}\]

Roughly, we may think of $\Phi(s_i)$ as the capacity for path $\pi$ to have future moves $m_j$, $j \geq i$ with  move length difference $\Delta(m_j) < 0$. 
The potential function $\Phi$ will be essential for analyzing the moves in \textsc{Forward-S}. However,  first we verify that the introduction of $\Phi$ does not affect our existing lower bounds on the move length differences of moves in \textsc{Backward}, \textsc{Zigzag}, \textsc{Forward-D}, or \textsc{Stationary}.

\begin{proposition}
Let $m_i$ be a move in $\normalfont \textsc{Backward}, \textsc{Zigzag},$ or $\normalfont \textsc{Forward-D}$. Then  $\widehat{\Delta}(m_i) = \Theta(\Delta(m_i))$. 
\label{prop:bad_moves_bad}
\end{proposition}
\begin{proof}
To prove this claim we need to upper bound $\Phi(s_{i+1}) - \Phi(s_i)$. Recall from our construction of $G_I$ that $S_I = [1, r_I/2] \times [1, y_I/2]$. Now observe that $\Phi(\cdot) \geq 0$, and $\Phi(s_i)$ is maximized when $s_{i, 1} = r_I/2$ and $s_{i, 2} = y_I/2 = n_I^{1/2}c^{-50} = \Theta(r_O^{2/3}c^{-108})$.
In this case,
\[
\Phi(s_i) =  \frac{s_{i, 2} \cdot ( i^*)^{-1} + s_{i, 1}}{r_I - c + i^*} \leq \frac{4s_{i, 2}}{r_I} = O(r_O^{2/3} \cdot c^{-210})
\]
Then  $\Phi(s_{i+1}) - \Phi(s_i) = O(r_O^{2/3}\cdot c^{-210})$. 
Now recall that by Lemmas \ref{lem:backward}, \ref{lem:zigzag}, and \ref{lem:forwardd}, $\Delta(m_i) = \Omega(r_O^{2/3}c^{-130})$. 
Then the claim immediately follows by taking $c$ to be sufficiently large.
\end{proof}

We can also lower bound the amortized move length difference of stationary moves.

\begin{lemma}
Let $m_i$ be a move in \normalfont \textsc{Stationary}.
\textit{Then}
$\widehat{\Delta}(m_i) \geq 0 $.
\label{lem:stat_move}
\end{lemma}
\begin{proof}
Observe that $d_i = 0$ because we stay in the same inner graph copy in this move, so $\vec{m}_i = \vec{0}$. Then $\widehat{\Delta}(m_i) = |m_i| - (\Phi(s_{i+1}) - \Phi(s_i))$. Let $s_i = (s_{i, 1}, s_{i, 2})$ and let $s_{i+1} = (s_{i+1, 1}, s_{i+1, 2})$.  Observe that $s_{i+1, 1} - s_{i, 1} \leq r_I/2$ since $s_i, s_{i+1} \in S_I$.   Let $m_i^I$ be the subpath of $m_i$ restricted to inner graph copy $G_I^{(i)}$. Note that $|m_i| \geq |m_i^I|$. 
Suppose that $|m_i^I| = d$ for some integer $d > 0$. 
Now note that for all vectors $\vec{v} = (v_1, v_2) \in W_I \cup -W_I$, we have that $v_1 \leq c^2$. Then by the construction of $G_I$ it follows that $s_{i+1, 2} - s_{i, 2} \leq dc^2$. 
Then using the fact that $r_I = c^{102}$ and $i^* \in [1, c]$,
\[\Phi(s_{i+1}) - \Phi(s_i) = \Phi(s_{i+1} - s_i)   \leq \frac{dc^2\cdot (i^*)^{-1} + r_I/2}{r_I - c + 1} \leq 2dc^{-100} + 2/3 \leq d\]

for sufficiently large $c$.  Then   $\widehat{\Delta}(m_i) \geq |m_i^I| - (\Phi(s_{i+1}) - \Phi(s_i)) \geq  0$  as claimed.
\end{proof}

Before lower bounding $\widehat{\Delta}(m_i)$ for moves $m_i$ in \textsc{Forward-S}, we first define a geometric object associated with $m_i$ that will be useful in our analysis.

\begin{definition}[Displacement vector]  
Let $m_i$ be a move  from input port $s_i$ in $G_I^{(i)}$ to input port $s_{i+1}$ in $G_I^{(i+1)}$, and let $t_i \in T_I$ be the output port incident to the subdivided path taken by $m_i$ to reach $s_{i+1}$. The displacement vector $\vec{\delta}_i = (\delta_{i, 1}, \delta_{i, 2})$ of $m_i$ is defined as
\[
\vec{\delta}_i = (t_i - s_i) - (t^* - s^*).
\]
\end{definition}
Observe that intuitively, $\vec{\delta}_i$ corresponds to the difference between the displacement from $s_i$ to $t_i$ and the displacement from $s^*$ to $t^*$ in $\mathbb{Z}^2$.
Despite being a purely geometric notion, $\vec{\delta}_i$ has two key properties related to the structure of graph $G$ that we will make use of in our analysis:
\begin{itemize}
\item \textbf{Update Property.} If $m_i$ is a move in \textsc{Forward-S}, then $s_{i+1}$ is uniquely determined by $s_i$ and $\vec{\delta}_i$.  
This property is formalized by Proposition \ref{prop:vec_height_preservation}.  
\item \textbf{Graph Distance Property.} The displacement vector $\vec{\delta}_i$ is tightly connected to $d_{G_I}(s_i, t_i)$, the graph distance from $s_i$ to $t_i$ in $G_I$. This property is formalized by Lemma \ref{prop:height_length}. 
\end{itemize}
As it turns out, these two properties will be sufficient to show that $\widehat{\Delta}(m_i) \geq 0$ for moves $m_i$ in \textsc{Forward-S}. We will proceed  by formalizing and proving the Update Property of $\vec{\delta}_i$.

\begin{proposition}[Update Property]
Let $m_i$ be a move from $s_i$ to $s_{i+1}$ in $\normalfont \textsc{Forward-S}$.  Then  
$$s_{i+1} = s_i + \vec{\delta}_i.$$ \label{prop:vec_height_preservation}
\end{proposition}
\begin{proof}
In move $m_i$ we take a subdivided path plugged into output port $t_i := (t_{i, 1}, t_{i, 2}) \in T_I$ in $G_I^{(i)}$.
This subdivided path corresponds to the outer graph vector $\vec{m}_i \in W_O$. 
By the construction of $\phi$, any such subdivided path will be plugged into an input port $s_{i+1}$ of $G_I^{(i+1)}$ such that $\phi(\vec{m}_i) = (s_{i+1}, t_i) \in P_I$.  By Proposition \ref{prop:P_I_x_diff}, for all $(s_{i+1}, t_i) \in P_I$,  $t_{i, 1} - s_{i+1, 1} = x_I - r_I/2$, where $s_{i+1, 1}$ denotes the first coordinate of $s_{i+1}$. Then $s_{i+1, 1} = t_{i, 1} - (x_I - r_I/2) = t_{i, 1} - (t_1^* - s_1^*)$, since $t_1^* - s_1^* = x_I - r_I/2$ by Proposition \ref{prop:P_I_x_diff}. Consequently,  $s_{i+1, 1} = s_{i, 1} + t_{i, 1} - s_{i, 1} - (t_1^* - s_1^*) = s_{i, 1} + \delta_{i, 1}$. 


Recall that since move $m_i$ is in \textsc{Forward-S},  vector $\vec{m}_i \in W_O$ is in the same stripe as $\vec{v}^* \in W_O$. Then by Proposition \ref{prop:phi}, the critical pair $\phi(\vec{m}_i) = (s_{i+1}, t_i) \in P_I$ has the same canonical inner graph vector in $W_I$ as the critical pair $\phi(\vec{v}^*) = (s^*, t^*) \in P_I$. 
Namely, critical pairs $\phi(\vec{m}_i)$ and $\phi(\vec{v}^*)$ share the canonical inner graph vector $\vec{v}_I^* = (\vec{v}_{I, 1}^*, \vec{v}_{I, 2}^*)  \in W_I$. 
Then since $t_{i, 1} - s_{i+1, 1} = t_1^* - s_1^*$ (by Proposition \ref{prop:P_I_x_diff}), we have that
\begin{align*}
    t_{i, 2} - s_{i+1, 2} = v_{I, 2}^* \cdot t_{i, 1} - s_{i+1, 1}) = (v_I^*)_y \cdot (t_x^* - s_x^{*}) = t_y^* - s_y^*
\end{align*}
by the construction of $P_I$. Then $s_y^{i+1} = t_y^i - (t_y^* - s_y^*) = s_y^i + \delta_y^i$. The claim follows. 
\end{proof}

The above proposition imposes a strong condition on the location of the next input port after a move in \textsc{Forward-S}.
We will use this proposition to argue about how the inner graph potential  $\Phi$ changes after a move in \textsc{Forward-S}. 
See Figure \ref{fig:update} for a visualization of the Update Property acting on moves in \textsc{Forward-S}.

\begin{figure}
    \centering
    \includegraphics[width=0.76\paperwidth]{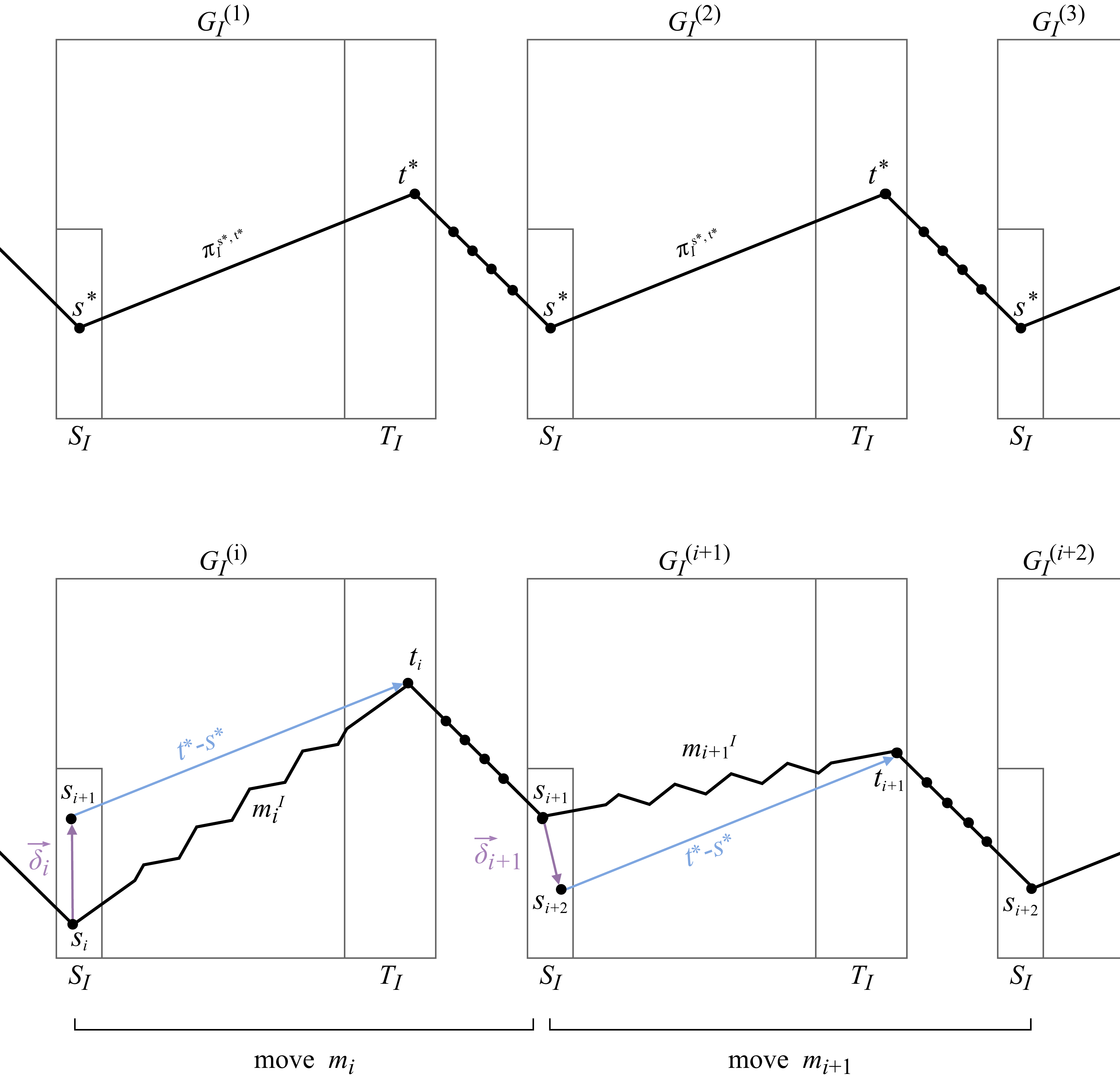}

\caption{\textit{Above,} two consecutive moves in $\pi^*$. The dotted lines between inner graphs represent subdivided paths in $G$. Note that all moves in $\pi^*$ are forward moves that take inner graph path $\pi_I^{s^*, t^*}$ and the subdivided path corresponding to vector $\vec{v}^* \in W_O$.  \textit{Below,} two consecutive moves $m_i$, $m_{i+1}$ belonging to \textsc{Forward-S} in path $\pi$.
Here $m_i^I$ denotes the restriction of $m_i$ to inner graph $G_I^{(i)}$.
Vectors $t^* - s^*$ and $\vec{\delta}_i$ are depicted in blue and purple respectively.    Note that by the Update Property, $s_{i+1}$ is uniquely determined by $s_i$ and $\vec{\delta}_i$. In particular, $s_{i+1} = s_i + \vec{\delta}_i$. Likewise, $s_{i+2} = s_{i+1} + \vec{\delta}_{i+1}$.  }

    \label{fig:update}
\end{figure}

\paragraph{Graph Distance Property.}  
Before stating the graph distance property formally, we will first attempt to convey an intuitive understanding of it. Because our graph $G_I$ is embedded in $\mathbb{Z}^2$ and the edges in $G_I$ correspond to vectors $W_I \subseteq \mathbb{Z}^2$, it is natural to imagine a  correspondence between distances in the Euclidean metric and distances in the graph metric of $G_I$. Roughly, we might presume that vertices $s_i, t_i$ in $G_I$ with a large  Euclidean distance $\|t_i - s_i\|$ in $\mathbb{Z}^2$ will have a large graph distance $d_{G_I}(s_i, t_i)$ in $G_I$. This understanding is approximately correct. 

Suppose that $\delta_x^i \geq 0$. Then 
 $t_x^i - s_x^i \geq t_x^* - s_x^*$,
so vertices $s_i, t_i$ have a larger horizontal displacement than vertices $s^*, t^*$.  
Since all our edges $\vec{w} = (w_x, w_y) \in W_I \cup -W_I$ in $G_I$ have horizontal displacement at most $w_x \leq r_I$, we can imagine that we must travel on more edges in $G_I$ to reach $t_i$ from $s_i$ than are needed to reach $t^*$ from $s^*$.  Likewise, if $\delta_x^i \leq 0$, then we can imagine that we may travel on fewer edges in $G_I$ to reach $t_i$ from $s_i$ than are needed to reach $t^*$ from $s^*$. 

Now suppose that  $\delta_y^i \geq 0$. Then 
 $t_y^i - s_y^i \geq t_y^* - s_y^*$,
so vertices $s_i, t_i$ have a larger vertical displacement than vertices $s^*, t^*$.  
Then it follows that every  $(s_i, t_i)$-path in $G_I$ must attain a larger vertical displacement than every $(s^*, t^*)$-path. Now note that by the construction of $W_I$, edges $\vec{w} = (w_x, w_y) \in W_I \cup -W_I$ with a larger vertical component $w_y$ have a smaller horizontal component $w_x$. Then we will need to travel on more of these edges with smaller horizontal components in order to reach $t_i$ from $s_i$. 
Thus we can imagine that more edges must be traversed in $G_I$ to reach $t_i$ from $s_i$ than are needed to reach $t^*$ from $s^*$. Likewise, if $\delta_y^i \leq 0$, then we can argue that fewer edges are needed to travel from $s_i$ to $t_i$, since we can travel on edges in $W_I$ with larger horizontal components. 

If we examine the moves $m_i$ and $m_{i+1}$ in Figure \ref{fig:update}, our intuition would suggest that $d_{G_I}(s_i, t_i) \gg d_{G_I}(s^*, t^*)$, since $t_y^i - s_y^i \gg t_y^* - s_y^*$. Likewise, our intuition would suggest that $d_{G_I}(s_{i+1}, t_{i+1}) \ll d_{G_I}(s^*, t^*)$, since $t_y^{i+1} - s_y^{i+1} \ll t_y^* - s_y^*$. To summarize, our (completely informal) argument suggests that the graph distance $d_{G_I}(s_i, t_i)$ increases with $\delta_x^i$ and $\delta_y^i$. We now present Lemma \ref{prop:height_length}, which states our precise formalization of the Graph Distance Property. 

\begin{lemma}[Graph Distance Property]
Let $s_i$ be an input port in $S_I$, and let $t_i$  be an output port in $T_I$, where $t_i$ is incident to a subdivided path in $G$. Suppose that $d_{G_I}(s_i, t_i) = |\pi_I^{s^*, t^*}| + d$, for some integer $d$. Then 
\[
\frac{\delta_y^i \cdot (i^*)^{-1} + \delta_x^i}{r_I - c + i^*} \leq d
\]
\label{prop:height_length}
\end{lemma}
\begin{proof}
This lemma requires a precise and technical analysis of graph distances in $G_I$, so we defer its proof to Section \ref{subsec:lemmaproof}. 
\end{proof}
Notice that the lower bound on $d_{G_I}(s_i, t_i)$ in Lemma \ref{prop:height_length} increases as $\delta_x^i$ and $\delta_y^i$ increase, matching our intuitive understanding of distances in $G_I$. 
Now that we have formally stated our two properties of $\vec{\delta}_i$, we are ready to lower bound the amortized move length difference $\widehat{\Delta}(m_i)$ for moves $m_i$ in \textsc{Forward-S}.

\begin{lemma}
Let $m_i$ be a move in $\normalfont \textsc{Forward-S}$. Then $\widehat{\Delta}(m_i) \geq 0$. \label{lem:forward_s_moves}
\end{lemma}
\begin{proof}
Let $m_i^I$ be the restriction of $m_i$ to inner graph copy $G_I^{(i)}$. Then $|m_i| = |m_i^I| + z$. Additionally,  $m_i^I$ must begin at some input port $s_i \in S_I$ and end at some output port $t_i \in T_I$, where $t_i$ is incident to a subdivided path. Note that $d_{G_I}(s_i, t_i) \leq |m_i^I|$. 
 Now recall that $L_{\pi^*} = \frac{|\pi_I^{s^*, t^*}| + z}{\|\vec{v}^*\|}$ and observe the following:

\begin{align*}
    \Delta(m_i) & = |m_i| - L_{\pi^*}  \cdot d_i  \\ 
    & = |m_i^I| + z - \frac{|\pi_I^{s^*, t^*}| + z}{\|\vec{v}^*\|}  \cdot d_i \\
   & \geq |m_i^I| + z - \frac{|\pi_I^{s^*, t^*}| + z}{\|\vec{v}^*\|}  \cdot \|\vec{v}^* \|  & \text{ by Property 3 of Lemma \ref{lem:CIS}} \\ 
   & \geq |m_i^I| - |\pi_I^{s^*, t^*}|
\end{align*}
\vspace{1mm}

Now let  $d_{G_I}(s_i, t_i) = |\pi_I^{s^*, t^*}| + d$ for some integer $d$. Then
\vspace{1mm}

\begin{align*}
    \widehat{\Delta}(m_i) & \geq |m_i^I| - |\pi_I^{s^*, t^*}|  - (\Phi(s_{i+1}) - \Phi(s_i))  \\
    & \geq d - (\Phi(s_{i+1}) - \Phi(s_i))   \\
    & = d - \frac{(s_y^{i+1} - s_y^i)(i^*)^{-1} + (s^{i+1}_x - s^i_x)}{r_I - c + i^*}
    \\
    & = d - \frac{\delta_y^i \cdot (i^*)^{-1} + \delta_x^i}{r_I - c + i^*} 
    &  \text{ by Proposition \ref{prop:vec_height_preservation}}
    \\
    & \geq d - d 
        &  \text{ by Lemma \ref{prop:height_length}}
    \\
    & \geq 0
\end{align*}
\end{proof}

With the exception of the proof of Lemma \ref{prop:height_length}, which has been deferred to Section \ref{subsec:lemmaproof}, we have established lower bounds for the amortized move difference $\widehat{\Delta}$ of all moves in the moveset. We will finish our lower bound argument in Section \ref{subsec:finishing}.

\subsection{Proving Lemma \ref{prop:height_length}  (Graph Distance Property)
\label{subsec:lemmaproof}
}


We will prove the following inequality, which is equivalent to the one stated in Lemma \ref{prop:height_length}:
\[ \delta_y^i  \leq  i^* (d   (r_I - c + i^*) - \delta_x^i)\]

Let $s_i := (s_x, s_y) \in S_I$, and let $t_i := (t_x, t_y) \in T_I$, where $t_i$ is incident to a subdivided path in $G$.  Let $\pi_i$ be a shortest $(s_i, t_i)$-path of length $\ell = |\pi_I^{s^*, t^*}| + d$ in $G_I$. Corresponding to $\pi_i$, we define the vector $\vec{v} = (v_x, v_y) = t_i - s_i$. Let $\vec{v}_1, \vec{v}_2, \dots, \vec{v}_{\ell}$ be the vectors corresponding to the edges of $\pi_i$ and observe that $\vec{v} = \sum_{i = 1}^{\ell} \vec{v}_i$. Furthermore, $\vec{v}_i \in W_I \cup -W_I$ for $i \in [1, \ell]$. Additionally, we define the vector $\vec{u} = (u_x, u_y) = \sum_{i=1}^{\ell} \vec{v}_I^*$.  

Observe that the quantity we want to upper bound is $\delta_y^i = (t_y - s_y) - (t_y^* - s_y^*) = v_y - (t_y^* - s_y^*)$.  We find that it easiest to separately upper bound $\delta_1 = v_y - u_y$ and $\delta_2 = u_y - (t_y^* - s_y^*)$. Then since $\delta_y^i = \delta_1 + \delta_2$, this will give us an upper bound of $\delta_y^i$ as desired. We first find the precise value of $\delta_2$.

\vspace{2mm}

\begin{proposition}
    \[
\delta_2 = d \cdot  \sum_{j=i^*}^c j
    \]
\end{proposition}
\begin{proof}
We find the precise value of $\delta_2$ by observing that since $\vec{u} = |\pi_I^{s^*, t^*}|\cdot \vec{v}_I^* + d\vec{v}_I^* = t^* - s^* +  d\vec{v}_I^*$,  it follows that
\[
\delta_2 = u_y - (t_y^* - s_y^*)  =  d \cdot \sum_{j=i^*}^c j  
\]
where the final equality follows from the fact that $\vec{v}_I^* = ((v_I^*)_x, (v_I^*)_y) =(r_I - c + i^*, \sum_{j = i^*}^c j)$.
\end{proof}

 For the rest of the proof, we will make use of the fact that  $\sum_{j = i}^c j = -i(i-1)/2 + c(c+1)/2 \geq 0$ for $i \in [1, c]$.  Before upper bounding $\delta_1$, we first observe that by Proposition \ref{prop:P_I_x_diff},
\begin{align*}
v_x - u_x &= v_x - (t^*_x - s^*_x) - d\cdot (v_I^*)_x\\
&= \delta_x^i - d \cdot (v_I^*)_x \\
& = \delta_x^i - d (r_I - c + i^*)
\end{align*}
Moreover, since $s_i \in S_I$, $1 \leq s_x \leq r_I/2$. Likewise, since $t_i \in T_I$, $x_I - r_I \leq t_i \leq x_I$. However, since $t_i$ is incident to a subdivided path, it follows by Proposition \ref{prop:P_I_x_diff} that the horizontal displacement between $t_i$ and some input port in $S_I$ is exactly $x_I - r_I/2$. Consequently, we obtain the stronger condition that $x_I - r_I/2 + 1 \leq t_i \leq x_I$. Then since $v_x = t_x - s_x$, it follows that $x_I - r_I \leq v_x \leq x_I$.  Finally, since $v_x - (t_x^* - s_x^*) = \delta_x^i$, we conclude that $$-r_I/2 \leq \delta_x^i \leq r_I/2.$$

Now we can upper bound $\delta_1$ by upper bounding the solution of  the following  optimization problem IP1:

\begin{mdframed}[align=center,userdefinedwidth=30em]
\textbf{IP1:} 
\begin{align*}
     \text{ maximize } \qquad &  v_y - u_y = \sum_{i=1}^{\ell} ( (v_i)_y - (v_I^*)_y) \\
    \text{subject to} \qquad  & v_x - u_x = \sum_{i=1}^{\ell} ((v_i)_x - (v_I^*)_x) \\
  &   \vec{v}_i = ((v_i)_x, (v_i)_y) \in W_I \cup -W_I \text{ for } i \in [1, \ell]
\end{align*}
\end{mdframed}

We now perform the following change of variables. Let
\[
\vec{q}_i := \vec{v}_i - \vec{v}_I^* \qquad \text{ for } i \in [1, \ell]
\]
Then the constraint that $\vec{v}_i \in W_I \cup -W_I$ becomes $\vec{q}_i \in Q_1 \cup Q_2$, where
$$
Q_1  := W_I - \vec{v}_I^* =  \left \{ \left(i - i^*,  \frac{-i(i-1) + i^*(i^* - 1)}{2}\right)  \bigm \vert i \in [1, c]  \right \} 
$$
and
\[
Q_2 := -W_I - \vec{v}_I^* =  \left \{
\left(-i - 2r_I + 2c - i^*, \frac{i(i - 1) + i^*(i^* - 1) }{2} -c(c+1) \right)  \bigm \vert i \in [1, c]
\right \}.
\]
Set $Q_1$ corresponds to the translated vectors from $W_I$ and $Q_2$ corresponds to the translated vectors from $-W_I$.  We now give a   linear relaxation of IP1:

\begin{mdframed}[align=center,userdefinedwidth=30em]
\textbf{LP1:}
\begin{align*}
    \text{maximize } \qquad  \delta_1' =  & \sum_{\vec{q}  \in Q_1 \cup Q_2} q_y w_{\vec{q}} \\
    \text{subject to} \qquad  & \sum_{\vec{q} \in Q_1 \cup Q_2} q_xw_{\vec{q}} =  v_x - u_x \\
    &
  &   w_{\vec{q}} \geq 0 \text{ for }  \vec{q}  \in Q_1 \cup Q_2
\end{align*}
\end{mdframed}

Note that $\delta_1' \geq \delta_1$, since
any feasible solution $\{\vec{v}_i\}_{i \in [1, \ell]}$ to IP1 can be converted to a feasible solution $\{w_{\vec{q}}\}_{\vec{q} \in Q_1 \cup Q_2}$ in LP1 with the same value. Thus it will be sufficient for our purposes to upper bound $\delta_1'$. To that end, we observe the following  useful fact about  solutions to  LP1. This fact, expressed in Proposition \ref{prop:optim},  follows straightforwardly from a convexity property of $Q_1$ and $Q_2$. 
We will defer the proof of Proposition \ref{prop:optim} to  Appendix \ref{app:optim}.  

\begin{proposition}
\label{prop:optim}
There exist optimal solutions to \text{\normalfont LP1} that are of the form
\[
\delta_1' = w_{\vec{q}_1}   \cdot (q_1)_y + w_{\vec{q}_2} \cdot (q_2)_y
\]
where $\vec{q}_1 = ((q_1)_x, (q_1)_y) \in Q_1$, $\vec{q}_2 = ((q_2)_x, (q_2)_y) \in Q_2$, and  $w_{\vec{q}_1}   \cdot (q_1)_x + w_{\vec{q}_2} \cdot (q_2)_x = v_x - u_x$.
\end{proposition}
\begin{proof}
Deferred to Appendix \ref{app:optim}.
\end{proof}

\noindent For the remainder of the proof, we split our analysis into three cases:
\begin{itemize}
    \item \textbf{Case 1:} $v_x - u_x > 0$
    \item \textbf{Case 2:} $v_x - u_x < 0$
    \item \textbf{Case 3:} $v_x - u_x = 0$
\end{itemize}

We begin with Case 1.

\vspace{3mm}

\begin{proposition}[Case 1]
    If $v_x - u_x > 0$, then $\delta_y^i  \leq  i^* (d   (r_I - c + i^*) - \delta_x^i)$.
    \label{prop:case1}
\end{proposition}

\begin{proof}




We may assume our solution is of the form $\delta_1' = w_{\vec{q}_1} \cdot  (q_1)_y + w_{\vec{q}_2} \cdot (q_2)_y$ subject to $\vec{q}_1 \in Q_1$, $\vec{q}_2 \in Q_2$, and $w_{\vec{q}_1}  \cdot (q_1)_x + w_{\vec{q}_2} \cdot (q_2)_x = v_x - u_x$. First note that $(q_2)_x < 0$ and $(q_2)_y < 0$ for all $\vec{q}_2 \in Q_2$, so it follows that $(q_1)_x > 0$. Then by the definition of $Q_1$, $(q_1)_y < 0$. Then the following feasible solution will be at least as good as our initial solution: 
$$\delta_1' = w'_{\vec{q}_1} \cdot (q_1)_y,  \text{\quad where \quad} w'_{\vec{q}_1} := (v_x - u_x) / (q_1)_x$$ 

\begin{itemize}
\item \textbf{Feasible:} Note that $v_x - u_x > 0$ and $(q_1)_x > 0$, so $w'_{\vec{q}_1} > 0$ as required. Additionally, the constraint $w'_{\vec{q}_1} \cdot (q_1)_x = v_x - u_x$ is satisfied. 
\item \textbf{Optimal:}  The difference between the value of our new solution and our old solution is 
\begin{align*}
 & w'_{\vec{q}_1} \cdot (q_1)_y - (w_{\vec{q}_1} \cdot  (q_1)_y + w_{\vec{q}_2} \cdot (q_2)_y) \\
& = (v_x - u_x) \cdot \frac{(q_1)_y}{(q_1)_x} -  (w_{\vec{q}_1} \cdot  (q_1)_y + w_{\vec{q}_2} \cdot (q_2)_y) \\
& = (w_{\vec{q}_1}  \cdot (q_1)_x + w_{\vec{q}_2} \cdot (q_2)_x)  \cdot \frac{(q_1)_y}{(q_1)_x} -  (w_{\vec{q}_1} \cdot  (q_1)_y + w_{\vec{q}_2} \cdot (q_2)_y) \\
& = w_{\vec{q}_2} \left( \frac{(q_2)_x (q_1)_y}{(q_1)_x} - (q_2)_y \right) \geq 0
\end{align*}
The inequality follows from the fact that $(q_1)_x > 0$, $(q_1)_y < 0$, $(q_2)_x < 0$, and $(q_2)_y < 0$ as noted earlier in the proof. We conclude our new solution is at least as good as the initial solution.
\end{itemize}

Consequently, we may assume that our  solution is of the form $\delta_1' = w_{\vec{q}_1} (q_1)_y$, where $w_{\vec{q}_1} (q_1)_x = v_x - u_x$. This term is maximized when we choose $\vec{q}_1 \in Q_1$ to be the vector  with the maximum $\frac{(q_1)_y}{(q_1)_x}$ ratio, where $(q_1)_x > 0$. It is straightforward to verify that vector $(1, -i^*) \in Q_1$ has the maximum such ratio $\frac{(q_1)_y}{(q_1)_x} = -i^*$.

We conclude that $\delta_1 \leq \delta_1' = \frac{(q_1)_y}{(q_1)_x} \cdot (v_x - u_x) =  -i^*(v_x - u_x)$.  Recall that $v_x - u_x = \delta_x^i - d (r_I - c + i^*)$ and $-r_I/ 2 \leq \delta_x^i \leq r_I/2$. Then since $0 \leq v_x - u_x $, it follows that $d \leq  \delta_x^i / (r_I - c + i^*) < 1$. Consequently, $d \leq 0$. Then
\begin{align*}
    \delta_y^i & = \delta_1 + \delta_2 \\
    & \leq -i^*(v_x - u_x) + d \cdot  \sum_{j=i^*}^c j \\
    & \leq i^*(u_x - v_x) = i^*(d (r_I - c + i^*) - \delta_x^i), 
\end{align*} as desired.
\end{proof}

\vspace{5mm}

\begin{proposition}[Case 2]
    If $v_x - u_x < 0$, then $\delta_y^i  \leq  i^* (d   (r_I - c + i^*) - \delta_x^i)$.
    \label{prop:case2}
\end{proposition}
\begin{proof}

Again we may assume our solution is of the form $\delta_1' = w_{\vec{q}_1} \cdot (q_1)_y + w_{\vec{q}_2} \cdot (q_2)_y$,  where $\vec{q}_1 \in Q_1$, $\vec{q}_2 \in Q_2$, and $w_{\vec{q}_1}  \cdot (q_1)_x + w_{\vec{q}_2} \cdot  (q_2)_x = v_x - u_x$. It can be seen that $\delta_1' \geq 0$ (e.g. by taking $w_{\vec{q}_2} = 0$ and $\vec{q}_1 = (-1, i^* - 1) \in Q_1$). Now note that $(q_2)_x < 0$ and $(q_2)_y < 0$ for all $\vec{q}_2 \in Q_2$. Additionally, note that if $(q_1)_x \geq 0$, then $(q_1)_y \leq 0$.
Consequently, if $(q_1)_x \geq 0$, then $\delta_1' < 0$, so we may assume $(q_1)_x <  0$.
Then since $(q_1)_x < 0$, it follows that $(q_1)_y > 0$ by the definition of $Q_1$.  On the other hand, we know that $(q_2)_y < 0$, so  the following feasible solution will be at least as good as our initial solution: 
$$\delta_1' = w'_{\vec{q}_1} \cdot (q_1)_y,  \text{\quad where \quad} w'_{\vec{q}_1} := (v_x - u_x) / (q_1)_x$$ 

\begin{itemize}
\item \textbf{Feasible:} Note that $v_x - u_x > 0$ and $(q_1)_x > 0$, so $w'_{\vec{q}_1} > 0$ as required. Additionally, the constraint $w'_{\vec{q}_1} \cdot (q_1)_x = v_x - u_x$ is satisfied. 
\item \textbf{Optimal:}  The difference between the value of our new solution and our old solution is 
\begin{align*}
 & w'_{\vec{q}_1} \cdot (q_1)_y - (w_{\vec{q}_1} \cdot  (q_1)_y + w_{\vec{q}_2} \cdot (q_2)_y) \\
& = (v_x - u_x) \cdot \frac{(q_1)_y}{(q_1)_x} -  (w_{\vec{q}_1} \cdot  (q_1)_y + w_{\vec{q}_2} \cdot (q_2)_y) \\
& = (w_{\vec{q}_1}  \cdot (q_1)_x + w_{\vec{q}_2} \cdot (q_2)_x)  \cdot \frac{(q_1)_y}{(q_1)_x} -  (w_{\vec{q}_1} \cdot  (q_1)_y + w_{\vec{q}_2} \cdot (q_2)_y) \\
& = w_{\vec{q}_2} \left( \frac{(q_2)_x (q_1)_y}{(q_1)_x} - (q_2)_y \right) \geq 0
\end{align*}
The inequality follows from the fact that $(q_1)_x < 0$, $(q_1)_y > 0$, $(q_2)_x < 0$, and $(q_2)_y < 0$ as noted earlier in the proof. We conclude our new solution is at least as good. 
\end{itemize}

Then we may assume our solution is of the form $w_{\vec{q}_1} (q_1)_y$, where $w_{\vec{q}_1} (q_1)_x = v_x - u_x$. This term is maximized when we choose $\vec{q}_1$ to be the vector  with the maximum $-\frac{(q_1)_y}{(q_1)_x}$ ratio, where $(q_1)_x < 0$. It is straightforward to verify that vector $(-1, i^* - 1) \in Q_1$ has the maximum such ratio $-\frac{(q_1)_y}{(q_1)_x} = i^* - 1$.

We conclude that $\delta_1 \leq \delta_1' = -\frac{(q_1)_y}{(q_1)_x} \cdot (v_x - u_x)  = -(i^* - 1)(v_x - u_x)$.
 Recall that $v_x - u_x = \delta_x^i - d (r_I - c + i^*)$ and $-r_I/ 2 \leq \delta_x^i \leq r_I/2$. Then since $0 \geq v_x - u_x $, it follows that $d \geq  \delta_x^i / (r_I - c + i^*) > -1$. Consequently, $d \geq 0$. 
Now note that if $d = 0$, then 
\[\delta_y^i = \delta_1 + \delta_2 = \delta_1 \leq (i^*-1)(u_x - v_x) \leq i^* (d   (r_I - c + i^*) - \delta_x^i)\]
where the final inequality follows from the fact that $u_x - v_x > 0$. Otherwise, if $d > 0$, then
\begin{align*}
    \delta_y^i = \delta_1 + \delta_2 & \leq (i^* - 1) \cdot (d   (r_I - c + i^*) - \delta_x^i) + d \cdot \sum_{j=i^*}^c j \\
    & \leq  i^*  \cdot (d   (r_I - c + i^*) - \delta_x^i) - (d   (r_I - c + i^*) - \delta_x^i) + dc^2 \\
    & \leq  i^*  \cdot (d   (r_I - c + i^*) - \delta_x^i) - d(r_I - c^2 - c) + \delta_x^i \\
    & \leq  i^*  \cdot  (d   (r_I - c + i^*) - \delta_x^i) - r_I/2 + \delta_x^i \\
    & \leq  i^*  \cdot  (d   (r_I - c + i^*) - \delta_x^i) - r_I/2 + r_I/2 \\
    & =  i^* (d   (r_I - c + i^*) - \delta_x^i)
\end{align*}
\end{proof}

\begin{proposition}[Case 3]
    If $v_x - u_x = 0$, then $\delta_y^i  \leq  i^* (d   (r_I - c + i^*) - \delta_x^i)$.
    \label{prop:case3}
\end{proposition}
\begin{proof}
Assume our solution is  of the form $\delta_1' = w_{\vec{q}_1}\cdot (q_1)_y + w_{\vec{q}_2}\cdot (q_2)_y$. By taking $w_{\vec{q}_1} = w_{\vec{q}_2} = 0$, we see that $\delta_1' \geq 0$. Moreover, if $w_{\vec{q}_2} > 0$, then $\delta_1' < 0$, so we may assume $w_{\vec{q}_2} = 0$. Then $\vec{q}_1 = (0, 0) \in Q_1$, so $\delta_1 \leq \delta_1' = 0$. Additionally, by the arguments in the previous two cases, $d = 0$.  
Then
\[
\delta_y^i = \delta_1 + \delta_2 = 0 = u_x - v_x  = d(r_I - c + i^*) - \delta_x^i \leq  i^* (d   (r_I - c + i^*) - \delta_x^i).
\]
\end{proof}

Lemma \ref{prop:height_length} follows immediately from Propositions \ref{prop:case1}, \ref{prop:case2}, \ref{prop:case3}.

\subsection{ Finishing the Lower Bound 
\label{subsec:finishing} }

Now that we have lower bounds for every category of move, we can begin reasoning about the difference in lengths between $\pi$ and $\pi^*$.

\begin{lemma}
Let $\pi$ be an $(s, t)$-path in $G$ that contains a move in $\normalfont  \textsc{Zigzag, Backward},$ or $\normalfont \textsc{Forward-D}$. Then $|\pi| - |\pi^*| = \Omega(r_O^{2/3}c^{-130})$.
\label{lem:bad_moves_bad}
\end{lemma}
\begin{proof}
Note that for every move $m$ in the move decomposition of $\pi$, $\widehat{\Delta}(m) \geq 0$. Moreover, since $\pi$ has a move $m'$ in $\normalfont  \textsc{Zigzag, Backward},$ or $\normalfont \textsc{Forward-D}$,  $\widehat{\Delta}(m') = \Omega(r_O^{2/3}c^{-130})$ by Proposition \ref{prop:bad_moves_bad}. Then  by Proposition \ref{prop:amort_diff}, $|\pi| - |\pi^*| = \Omega(r_O^{2/3}c^{-130})$.
\end{proof}

\begin{lemma}
\label{lemma:bad_moves}
Let $\pi$ be an $(s, t)$-path in $G$ with moves only in $\normalfont \textsc{Stationary}$ or $\normalfont \textsc{Forward-S}$ that takes a subdivided path not in $\pi^*$. Then $|\pi| - |\pi^*| = \Omega(r_O^{2/3}c^{-111})$. 
\end{lemma}
\begin{proof}
It will be useful to split the edges in $\pi$ and $\pi^*$ into two categories, the edges in the inner graph and the  edges in the outer graph. Let $\ell_I^{\pi}$ and $\ell_O^{\pi}$ denote the number of edges in the inner graph (respectively outer graph) in $\pi$, so $|\pi| = \ell_I^{\pi} + \ell_O^{\pi}$. Define $\ell_I^{\pi^*}$ and $\ell_O^{\pi^*}$ similarly.  Now since $\pi$ takes a subdivided path not in $\pi^*$, by part 5 of Lemma \ref{lemma:outer_graph} we know that $\pi$ takes at least one more subdivided path than $\pi^*$, so $\ell_O^{\pi} - \ell_O^{\pi^*} \geq z$, where $z$ is the length of a subdivided path in $G$.  Additionally, by our amortized analysis, we know that $|\pi| - |\pi^*| \geq 0$. However, we claim that our amortized analysis proves the stronger statement that $\ell_I^{\pi} - \ell_I^{\pi^*} \geq 0$.

Let $m_1, m_2, \dots, m_k$ be the move decomposition of $\pi$. For $i \in [1, k]$, let $m_i^I$  denote the restriction of move $m_i$ to inner graph copy $G_I^{(i)}$.
Note that $\sum_i |m_i^I| = \ell_I^{\pi}$. 
We now define an amortized \textit{inner graph} move length difference function $\widehat{\Delta}_I(m_i)$ for moves $m_i$, $i \in [1, k]$. 
If $m_i$ is a move in \textsc{Stationary}, then let $\widehat{\Delta}_I(m_i) = |m_i^I|  - (\Phi(s_{i+1}) - \Phi(s_i))$. Likewise, if $m_i$ is a move in \textsc{Forward-S}, then let $\widehat{\Delta}_I(m_i) = |m_i^I|  - |\pi_I^{s^*, t^*}|  - (\Phi(s_{i+1}) - \Phi(s_i))$. 

Suppose that $\pi$ contains $q \in [1, k]$ moves in \textsc{Forward-S}. Now note that every move in \textsc{Forward-S} can be interpreted as travelling on an edge in $G_O$. Additionally, path $\pi^*$ is composed exclusively of moves in \textsc{Forward-S}. Then by the unique shortest path property of $G_O$ (part 5 of Lemma \ref{lemma:outer_graph}), path $\pi^*$ has at most $q$ moves, so $\ell_I^{\pi^*} \leq q \cdot |\pi_I^{s^*, t^*}|$. It follows that
\[
\sum_i \widehat{\Delta}_I(m_i) = \sum_i |m_i^I| - q\cdot |\pi_I^{s^*, t^*}| - (\Phi(t) - \Phi(s)) = \ell_I^{\pi} - q\cdot |\pi_I^{s^*, t^*}| \leq  \ell_I^{\pi} - \ell_I^{\pi^*}
\]
Additionally, observe that the proofs of Lemma \ref{lem:stat_move} and Lemma  \ref{lem:forward_s_moves} establish that $\widehat{\Delta}_I(m_i) \geq 0$ if $m_i$ is in \textsc{Stationary} or \textsc{Forward-S}. Then it immediately follows that $\ell_I^{\pi} - \ell_I^{\pi^*} \geq \sum_i \widehat{\Delta}_I(m_i) \geq 0$. 
Consequently, $|\pi| - |\pi^*| = (\ell_I^{\pi} + \ell_O^{\pi}) - (\ell_I^{\pi^*} + \ell_O^{\pi^*}) \geq z =  \Omega(r_O^{2/3}c^{-111})$. 
\end{proof}

The following lemma is immediate from Lemma \ref{lem:bad_moves_bad} and Lemma \ref{lemma:bad_moves}.

\begin{lemma}
Let $\pi$ be an $(s, t)$-path in $G$ that takes a subdivided path not in $\pi^*$. Then $|\pi| - |\pi^*| = \Omega(r_O^{2/3}c^{-130})$.
\label{lem:gap_lem}
\end{lemma}

We can now prove our main result. 

\begin{theorem}
For any sufficiently large parameter $c_0$,  there are infinitely many $n$ for which there is an $n$-vertex graph $G$ such that any spanner of $G$ with at most $c_0n$ edges has additive distortion $+\Omega(n^{1/7}c_0^{-80})$.
\label{thm:spanner_lb}
\end{theorem}
\begin{proof}
We are given a sufficiently large parameter $c_0 > 0$.
Then we will choose construction parameter $c = \Theta(c_0)$ for our infinite family of graphs $G$. Recall that by Lemma \ref{lem:final_graph_2}, every graph in our family on $n$ nodes has $m = \Theta(cn)$ edges. We choose $c$ to be sufficiently large so that for every graph $G$ on $n$ vertices and $m$ edges in our family, $\frac{c_0n}{m} < 1/2$.
Now consider any spanner $H$ of graph $G$
with at most $c_0n$ edges. Graph $H$ will contain at most half of the edges of $G$. Then because the canonical paths of $G$ are a partition of the edges of $G$, some canonical $(s, t)$-path $\pi^*$ has at most half of its edges in $H$, for some $(s, t) \in P$. 

Now fix a shortest $(s, t)$-path $\pi$ in $H$. If path $\pi$ takes a subdivided path in $G$ not in path $\pi^*$, then by Lemma \ref{lem:gap_lem}, $|\pi| - |\pi^*| = \Omega(r_O^{2/3}c^{-130})$. Else, path $\pi$ travels through exactly the inner graph copies that path $\pi^*$ travels through. Moreover, all the subdivided paths of $\pi^*$ must be in $H$.
Then at least half of the inner graph subpaths of $\pi^*$ are missing an edge in $H$. Now fix an inner graph copy $G_I$ that $\pi$ and $\pi^*$ pass through in $G$, and let $\pi_I$ and $\pi_I^*$ be the subpaths of these paths restricted to graph $G_I$. Recall that $\pi_I^*$ is a unique shortest path between its endpoints in $G_I$. Then if path $\pi_I^*$ is missing an edge in $H$, then $\pi_I \neq \pi_I^*$, so $|\pi_I| - |\pi_I^*| \geq 1$. Finally, observe that paths $\pi^*$ and $\pi$ pass through $\Omega(n_O^{1/2}r_O^{-1})$ inner graphs. It follows that $|\pi| -|\pi^*| = \Omega(n_O^{1/2}r_O^{-1})$.

Then balancing our parameters so that $r_O^{2/3}c^{-130} = \Theta(n_O^{1/2}r_O^{-1})$  gives us $r_O = \Theta(n_O^{3/10}c^{78})$. We note that $r_O \leq \frac{x_O}{4}$, so our construction requirements are satisfied. 
By Lemma \ref{lem:final_graph}, $G$ will have $n = \Theta(n_O \cdot r_O^{4/3}c^{-116}) = \Theta(n_O^{7/5}c^{-12})$ vertices. By the argument in the previous paragraph, $|\pi| - |\pi^*| =  \Omega(n_O^{1/5}c^{-78}) = \Omega(n^{1/7}c^{-80})$.  So $H$ will have additive distortion of at least $+\Omega(n^{1/7}c_0^{-80})$.  
\end{proof}

By choosing a different value for $r_O$, we are able to obtain new lower bounds against pairwise additive spanners.


\begin{theorem}
For any sufficiently large parameter $c_0$, there are infinitely many $n$ for which there is an $n$-vertex graph $G$ and a set $P$ of $p = \Theta_{c_0}(n^{1/2})$ demand pairs such that any pairwise spanner of $(G, P)$ with at most $c_0n$ edges has additive distortion $+c_0$. 
\end{theorem}
\begin{proof}
Given $c_0 > 0$, we choose construction parameter $c = \Theta(c_0)$ so that for every graph $G$ on $n$ vertices and $m$ edges in our infinite family of graphs, $\frac{c_0n}{m} < 1/2$. We define our demand pairs to be the set of critical pairs $P$ in our construction.
Let $H$ be a pairwise spanner of $(G, P)$ with at most $c_0n$ edges. By an argument identical to that of Theorem \ref{thm:spanner_lb}, it follows that $H$ has additive distortion at least $+k = \min\{r_O^{2/3}c^{-130}, n_O^{1/2}r_O^{-1}\}$.

 Now instead of choosing $r_O$ to grow polynomially with $n$ as in Theorem \ref{thm:spanner_lb}, we will let $r_O = \Theta_{c_0}(1)$. Moreover, we will require that $r_O > c_0^{300}$. Note that there exists infinitely many valid choices of $r_O$ satisfying these criteria. Then for sufficiently large $c_0$ and sufficiently large $n$ relative to $c_0$, it follows that $+k = \min\{r_O^{2/3}c^{-130}, n_O^{1/2}r_O^{-1}\} \geq c_0$. Additionally, by Lemma \ref{lem:final_graph_2}, $|P| = \Theta(r_O^{5/3} \cdot c^{-5} \cdot n_O^{1/2}) = \Theta_{c_0}(n_O^{1/2})$, and $n = \Theta(n_O \cdot r_O^{4/3}c^{-116}) =  \Theta_{c_0}(n_O)$, so we conclude that $|P| = \Theta_{c_0}(n^{1/2})$. 
\end{proof}

\section*{Acknowledgements}
The authors would like to thank Shang-En Huang for helpful discussions about lower bound graphs.

\bibliographystyle{plain}
\bibliography{refs}

\begin{thebibliography}{10}

\bibitem{abboud2016error}
Amir Abboud and Greg Bodwin.
\newblock Error amplification for pairwise spanner lower bounds.
\newblock In {\em Proceedings of the twenty-seventh annual ACM-SIAM symposium
  on Discrete algorithms}, pages 841--854. SIAM, 2016.

\bibitem{AB16soda}
Amir Abboud and Greg Bodwin.
\newblock Error amplification for pairwise spanner lower bounds.
\newblock In {\em Proceedings of the 27th Annual ACM-SIAM Symposium on Discrete
  Algorithms (SODA)}, pages 841--854. Society for Industrial and Applied
  Mathematics, 2016.

\bibitem{AB17jacm}
Amir Abboud and Greg Bodwin.
\newblock The 4/3 additive spanner exponent is tight.
\newblock {\em Journal of the ACM (JACM)}, 64(4):28:1--28:14, 2017.

\bibitem{ABP17}
Amir Abboud, Greg Bodwin, and Seth Pettie.
\newblock A hierarchy of lower bounds for sublinear additive spanners.
\newblock In {\em Proceedings of the 28th Annual ACM-SIAM Symposium on Discrete
  Algorithms (SODA)}, pages 568--576. Society for Industrial and Applied
  Mathematics, 2017.

\bibitem{ABSHJKS20}
Reyan Ahmed, Greg Bodwin, Faryad~Darabi Sahneh, Keaton Hamm, Mohammad
  Javad~Latifi Jebelli, Stephen Kobourov, and Richard Spence.
\newblock Graph spanners: A tutorial review.
\newblock {\em Computer Science Review}, 37:100253, 2020.

\bibitem{ACIM99}
Donald Aingworth, Chandra Chekuri, Piotr Indyk, and Rajeev Motwani.
\newblock Fast estimation of diameter and shortest paths (without matrix
  multiplication).
\newblock {\em SIAM Journal on Computing}, 28(4):1167--1181, 1999.

\bibitem{AlDhalaan21}
Bandar Al-Dhalaan.
\newblock Fast construction of 4-additive spanners.
\newblock {\em arXiv preprint arXiv:2106.07152}, 2021.

\bibitem{Alon02}
Noga Alon.
\newblock Testing subgraphs in large graphs.
\newblock {\em Random Structures \& Algorithms}, 21(3-4):359--370, 2002.

\bibitem{ADDJS93}
Ingo Alth{\"o}fer, Gautam Das, David Dobkin, Deborah Joseph, and Jos{\'e}
  Soares.
\newblock On sparse spanners of weighted graphs.
\newblock {\em Discrete \& Computational Geometry}, 9(1):81--100, 1993.

\bibitem{Awerbuch85}
Baruch Awerbuch.
\newblock Complexity of network synchronization.
\newblock {\em Journal of the ACM (JACM)}, 32(4):804--823, 1985.

\bibitem{balog1991convex}
Antal Balog and Imre B{\'a}r{\'a}ny.
\newblock On the convex hull of the integer points in a disc.
\newblock In {\em Proceedings of the Seventh Annual Symposium on Computational
  Geometry}, pages 162--165, 1991.

\bibitem{BL98}
Imre B{\'a}r{\'a}ny and David~G Larman.
\newblock The convex hull of the integer points in a large ball.
\newblock {\em Mathematische Annalen}, 312(1):167--181, 1998.

\bibitem{BKMP10}
Surender Baswana, Telikepalli Kavitha, Kurt Mehlhorn, and Seth Pettie.
\newblock Additive spanners and ($\alpha$, $\beta$)-spanners.
\newblock {\em ACM Transactions on Algorithms (TALG)}, 7(1):5, 2010.

\bibitem{BCLR86}
Sandeep Bhatt, Fan Chung, Tom Leighton, and Arnold Rosenberg.
\newblock Optimal simulations of tree machines.
\newblock In {\em 27th Annual Symposium on Foundations of Computer Science
  (sfcs 1986)}, pages 274--282. IEEE, 1986.

\bibitem{Bodwin21}
Greg Bodwin.
\newblock New results on linear size distance preservers.
\newblock {\em SIAM Journal on Computing}, 50(2):662--673, 2021.

\bibitem{BV15}
Greg Bodwin and Virginia~Vassilevska Williams.
\newblock Very sparse additive spanners and emulators.
\newblock In {\em Proceedings of the 2015 Conference on Innovations in
  Theoretical Computer Science (ITCS)}, pages 377--382. ACM, 2015.

\bibitem{BV16}
Greg Bodwin and Virginia~Vassilevska Williams.
\newblock Better distance preservers and additive spanners.
\newblock In {\em Proceedings of the 27th Annual ACM-SIAM Symposium on Discrete
  Algorithms (SODA)}, pages 855--872. Society for Industrial and Applied
  Mathematics, 2016.

\bibitem{BCE05}
B{\'e}la Bollob{\'a}s, Don Coppersmith, and Michael Elkin.
\newblock Sparse distance preservers and additive spanners.
\newblock {\em SIAM Journal on Discrete Mathematics}, 19(4):1029--1055, 2005.

\bibitem{Cai94}
Leizhen Cai.
\newblock Np-completeness of minimum spanner problems.
\newblock {\em Discrete Applied Mathematics}, 48(2):187--194, 1994.

\bibitem{Chechik13soda}
Shiri Chechik.
\newblock New additive spanners.
\newblock In {\em Proceedings of the 24th Annual ACM-SIAM Symposium on Discrete
  Algorithms (SODA)}, pages 498--512. SIAM, 2013.

\bibitem{CKRS90}
Jason Cong, Andrew~B Kahng, Gabriel Robins, Majid Sarrafzadeh, and CK~Wong.
\newblock {\em Performance-driven global routing for cell based IC's}.
\newblock University of California (Los Angeles). Computer Science Department,
  1990.

\bibitem{CKRS92}
Jason Cong, Andrew~B Kahng, Gabriel Robins, Majid Sarrafzadeh, and CK~Wong.
\newblock Provably good algorithms for performance-driven global routing.
\newblock {\em Proc. of 5th ISCAS}, pages 2240--2243, 1992.

\bibitem{CE06}
Don Coppersmith and Michael Elkin.
\newblock Sparse sourcewise and pairwise distance preservers.
\newblock {\em SIAM Journal on Discrete Mathematics}, 20(2):463--501, 2006.

\bibitem{cygan2013pairwise}
Marek Cygan, Fabrizio Grandoni, and Telikepalli Kavitha.
\newblock On pairwise spanners.
\newblock In {\em 30th International Symposium on Theoretical Aspects of
  Computer Science}, page 209, 2013.

\bibitem{EP04}
Michael Elkin and David Peleg.
\newblock (1+$\varepsilon$,$\beta$)-spanner constructions for general graphs.
\newblock {\em SIAM Journal on Computing}, 33(3):608--631, 2004.

\bibitem{Hesse03}
William Hesse.
\newblock Directed graphs requiring large numbers of shortcuts.
\newblock In {\em Proceedings of the fourteenth annual ACM-SIAM symposium on
  Discrete algorithms}, pages 665--669. Society for Industrial and Applied
  Mathematics, 2003.

\bibitem{HP18}
Shang-En Huang and Seth Pettie.
\newblock {Lower Bounds on Sparse Spanners, Emulators, and Diameter-reducing
  shortcuts}.
\newblock In David Eppstein, editor, {\em 16th Scandinavian Symposium and
  Workshops on Algorithm Theory (SWAT 2018)}, volume 101 of {\em Leibniz
  International Proceedings in Informatics (LIPIcs)}, pages 26:1--26:12,
  Dagstuhl, Germany, 2018. Schloss Dagstuhl--Leibniz-Zentrum fuer Informatik.

\bibitem{kavitha2017new}
Telikepalli Kavitha.
\newblock New pairwise spanners.
\newblock {\em Theory of Computing Systems}, 61(4):1011--1036, 2017.

\bibitem{kavitha2013small}
Telikepalli Kavitha and Nithin~M Varma.
\newblock Small stretch pairwise spanners.
\newblock In {\em International Colloquium on Automata, Languages, and
  Programming}, pages 601--612. Springer, 2013.

\bibitem{doi:10.1137/140953927}
Telikepalli Kavitha and Nithin~M. Varma.
\newblock Small stretch pairwise spanners and approximate \$d\$-preservers.
\newblock {\em SIAM Journal on Discrete Mathematics}, 29(4):2239--2254, 2015.

\bibitem{Knudsen14}
Mathias B{\ae}k~Tejs Knudsen.
\newblock Additive spanners: A simple construction.
\newblock In {\em Scandinavian Workshop on Algorithm Theory}, pages 277--281.
  Springer, 2014.

\bibitem{LS91}
Arthur~L Liestman and Thomas~C Shermer.
\newblock Additive spanners for hypercubes.
\newblock {\em Parallel Processing Letters}, 1(01):35--42, 1991.

\bibitem{Lu19}
Kevin Lu.
\newblock {\em New methods for approximating shortest paths}.
\newblock PhD thesis, Massachusetts Institute of Technology, 2019.

\bibitem{LVWX22}
Kevin Lu, Virginia~Vassilevska Williams, Nicole Wein, and Zixuan Xu.
\newblock Better lower bounds for shortcut sets and additive spanners via an
  improved alternation product.
\newblock In {\em Proceedings of the 2022 Annual ACM-SIAM Symposium on Discrete
  Algorithms (SODA)}, pages 3311--3331. SIAM, 2022.

\bibitem{PU89sicomp}
David Peleg and Jeffrey Ullman.
\newblock An optimal synchronizer for the hypercube.
\newblock {\em SIAM Journal on Computing (SICOMP)}, 18(4):740–--747, 1989.

\bibitem{PU89jacm}
David Peleg and Eli Upfal.
\newblock A trade-off between space and efficiency for routing tables.
\newblock {\em Journal of the ACM (JACM)}, 36(3):510--530, 1989.

\bibitem{Pettie09}
Seth Pettie.
\newblock Low distortion spanners.
\newblock {\em ACM Transactions on Algorithms (TALG)}, 6(1):7, 2009.

\bibitem{SCRS01}
F~Sibel Salman, Joseph Cheriyan, Ramamoorthi Ravi, and Sairam Subramanian.
\newblock Approximating the single-sink link-installation problem in network
  design.
\newblock {\em SIAM Journal on Optimization}, 11(3):595--610, 2001.

\bibitem{TZ06}
Mikkel Thorup and Uri Zwick.
\newblock Spanners and emulators with sublinear distance errors.
\newblock In {\em Proceedings of the 17th Annual ACM-SIAM Symposium on Discrete
  Algorithms (SODA)}, pages 802--809. Society for Industrial and Applied
  Mathematics, 2006.

\bibitem{Woodruff06}
David~P Woodruff.
\newblock Lower bounds for additive spanners, emulators, and more.
\newblock In {\em Foundations of Computer Science, 2006. FOCS'06. 47th Annual
  IEEE Symposium on}, pages 389--398. IEEE, 2006.

\bibitem{Woodruff10}
David~P Woodruff.
\newblock Additive spanners in nearly quadratic time.
\newblock In {\em International Colloquium on Automata, Languages, and
  Programming}, pages 463--474. Springer, 2010.

\end{thebibliography}

\begin{appendix}
\section{Proof of Lemma \ref{lem:CIS} \label{app:CIS}}

Our proof of Lemma \ref{lem:CIS} begins with a result of Balog and B{\'a}r{\'a}ny which we restate in our specific context below:

\begin{lemma}[\cite{balog1991convex}, Theorem 1] \label{lem:barany}
For $r \ge 0$, there exists a strongly convex set $W(r)$ of integer vectors in $\mathbb{Z}^2$ of size $|W(r)| = \Theta(r^{2/3})$ such that for all $\vec{v} \in W(r)$, $r - r^{-1/3} \leq \|\vec{v}\| \leq r$.
\end{lemma}

The fact that all vectors $\vec{v} \in W(r)$ satisfy $r - r^{-1/3} \leq \|\vec{v}\| \leq r$ is not stated explicitly in \cite{balog1991convex}. However, it follows directly from the lower bound argument in Section 3 of \cite{balog1991convex}, specifically the assignment of parameter $\Delta$.
The fact that $W(r)$ is \emph{strongly} convex is also technically not stated in \cite{balog1991convex}, but it follows immediately from symmetry of the construction; that is, if $\vec{v} \in W(r)$ then $-\vec{v} \in W(r)$.
We will now remove some of the vectors from the strongly convex set $W(r)$ from Lemma \ref{lem:barany} in order to obtain a strongly convex set satisfying Lemma \ref{lem:CIS}. 
Note that any subset of $W(r)$ will be a strongly convex set satisfying property 1 of Lemma \ref{lem:CIS}, so it remains to obtain the latter two properties.

\paragraph{Property 2.} Next, we try to enforce the second property of Lemma \ref{lem:CIS}: for any circular sector $S$ with inner angle $\psi$, there are only $O(\psi \cdot r^{2/3})$ vectors in $W(r) \cap S$.
To do so: restrict $W(r)$ to be the vectors in just one quadrant of $\rr^2$ (throwing away a constant fraction of these vectors), and let $\vec{v_1}, \vec{v_2}, \dots, \vec{v_k}$ be the vectors in $W(r)$ ordered counterclockwise about the origin.
Let $\vec{u_i} = \vec{v}_{i+1} - \vec{v_i}$ for $i \in [1, k]$ (where indices are taken mod $k$).
Observe that $\vec{u_i} \neq \vec{u_j}$ for all $i \neq j$ by the convexity of $W(r)$. Moreover, by integrality of the vectors in $W(r)$, we have that all $\vec{u_i}$ have integer coordinates.
Now observe that there are only $4(\ell+1)$ distinct integer vectors $(x, y)$ such that $|x| + |y| = \ell$, namely, the vectors
$$(i, \ell-i), (-i, \ell-i), (i, -(\ell-i)), (-i, -(\ell-i)) \qquad \text{for } i \in [0, \ell].$$
Thus, the number of these vectors $\vec{u}_i$ of magnitude $\|\vec{u}_i\| \leq m$ is at most
$$\sum_{\ell=0}^{m}4(\ell+1) \leq 8m^2.$$ 

We will throw away the vectors $\vec{v_i}$ in $W(r)$ that are close together.
Let $|W(r)| \geq \alpha r^{2/3}$ where $\alpha >0$ is some constant. We define $W'(r)$ as follows:
\[
W'(r) = \left\{ \vec{v}_i \in W(r)   \mid   \|\vec{u_i}\| > \frac{\alpha^{1/2}}{4} \cdot r^{1/3} \right\} 
\]
Since each $\vec{u_i}$ is a distinct integer vector, by our previous observations there can be no more than $\alpha /2\cdot r^{2/3}$ vectors $\vec{u_i}$ with magnitude $\|\vec{u_i}\| \leq \alpha^{1/2}/4 \cdot r^{1/3}$. Thus
$$\left|W'(r)\right| \geq \frac{\alpha}{2} \cdot r^{2/3}$$
as desired. Now we claim that $W'(r)$ has property 2.
For notational convenience, we will redefine $\vec{v}_1, \dots, \vec{v}_k$ and vectors $\{\vec{u}_i\}$ as before, over the \emph{surviving} vectors in $W'(r)$.
Since all vectors in $W'(r)$ lie in the same quadrant of $\rr^2$, all vectors $\{\vec{u}_i\}$ have the same signs for their respective coordinates, and it follows that for all $i$ we have
$$\|\vec{u_i}\| > \frac{\alpha^{1/2}}{4} \cdot r^{1/3}.$$
Let $S$ be a circular sector with inner angle $\psi$, and let $\vec{v}_i, \vec{v}_{i+1}, \dots, \vec{v}_{i + j}$ be the $j+1$ vectors in $W'(r) \cap S$. 
Now consider the convex hull $CH(\vec{v_i}, \dots, \vec{v}_{i+j})$ of the endpoints of vectors $\vec{v_i}, \dots, \vec{v}_{i+j}$.
Observe that the boundary of this convex hull has length
$$\sum_{\nu=i}^{i+j-1} \|\vec{u}_{\nu}\| + \| \vec{v}_{i+j} - \vec{v}_i \|.$$
Furthermore, if $a_S$ denotes the circular arc of radius $r$ spanning the circular sector $S$, then the convex hull $CH(a_S, \vec{v_i}, \vec{v}_{i+j})$ encloses $CH(\vec{v_i}, \dots, \vec{v}_{i+j})$. 
The length of the boundary of $CH(a_S, \vec{v_i}, \vec{v}_{i+j})$ is dominated by the arc $a_S$, so it is $O(r\psi)$.
We thus have
\begin{gather*}
\Omega\left(j \cdot r^{1/3}\right) \leq  \sum_{\nu=i}^{i+j-1} \|\vec{u}_{\nu}\| + \| \vec{v}_{i+j} - \vec{v}_i\| \leq O(r\psi)
\end{gather*}
so we conclude that $j \leq O\left(\psi \cdot r^{2/3}\right)$.

\paragraph{Property 3.}
Lastly, we enforce the property that for each vector $\vec{v}$, the vector in our strongly convex set with the largest magnitude scalar projection onto $\vec{v}$ is $\vec{v}$ itself.
To do so, we will pass from $W'(r)$ to $W''(r)$.
Note that any subset of $W'(r)$ will be a strongly convex set and satisfy properties 1 and 2 of Lemma \ref{lem:CIS}.

First we overview our strategy to construct $W''(r)$.
We repeatedly select some $\vec{v} \in W'(r)$ and add it to $W''(r)$.
Let $\ell_v$ be the chord of the circle of radius $r$ that intersects the endpoint of $\vec{v}$ and is perpendicular to it.
If there are vectors $\vec{u} \in W'(r)$ with scalar projection $\proj_{\vec{v}}\vec{u} \geq \| \vec{v} \|$, then they lie on the side of $\ell_v$ opposite the origin. 
We throw away all such vectors from $W'(r)$ before proceeding, and selecting one of the surviving vectors in $W'(r)$ as our next choice of $\vec{v}$.

Let us bound the total number of vectors discarded in this way.
Using the fact that $r - r^{-1/3} \le \|\vec{v}\| \le r$, a straightforward application of the Pythagorean theorem shows that the chord $\ell_v$ has length $O(r^{1/3})$.
Thus, the angle $\psi_v$ of the circular sector spanning $\ell_v$ satisfies $\sin(\psi_v) = O(r^{-2/3})$, and thus $\psi_v = O(r^{-2/3})$.
Applying property $2$ of Lemma \ref{lem:CIS}, this circular sector contains only $O(1)$ vectors from $W'(r)$.
Thus, we discard only $O(1)$ vectors from $W'(r)$ for each vector that we add to $W''(r)$, and so we discard only a constant fraction of the vectors in $W'(r)$ in total.

So far, this construction guarantees that, for any vector $\vec{v}$ added to $W''(r)$, no \emph{future} vector $\vec{v}'$ added to $W''(r)$ has larger projection onto $\vec{v}$ than $\vec{v}$ itself.
We still need to handle the opposite dependency, in which a \emph{past} vector $\vec{v}'$ added to $W''(r)$ has larger projection onto $\vec{v}$ than $\vec{v}$ itself.
In fact, this is easily handled by considering the vectors in $W'(r)$ in a specific order, rather than choosing an arbitrary $\vec{v} \in W'(r)$ in each round.
Notice that this problematic dependence only occurs in the case when $\|\vec{v}'\| > \|\vec{v}\|$, and so $\vec{v}$ escapes the discarded sector for $\vec{v}'$, but $\vec{v}'$ lies in the discarded sector for $\vec{v}$.
We therefore select the vector $\vec{v} \in W'(r)$ with \emph{minimum length} $\|\vec{v}\|$ in each round (breaking ties arbitrarily).
This avoid conflicts as described, since we would then consider $\vec{v}$ first, and discard $\vec{v}'$.

\section{Proof of Proposition \ref{prop:optim} \label{app:optim}}

We must show that any optimal solution to our relaxed optimization problem is of form $\delta_1' = w_{\vec{q}_1}   \cdot (q_1)_y + w_{\vec{q}_2} \cdot (q_2)_y$, where $\vec{q}_1 = ((q_1)_x, (q_1)_y) \in Q_1$, $\vec{q}_2 = ((q_2)_x, (q_2)_y) \in Q_2$, and  $w_{\vec{q}_1}   \cdot (q_1)_x + w_{\vec{q}_2} \cdot (q_2)_x = v_x - u_x$. We can accomplish this by showing that any feasible solution in general form can be converted into a feasible solution in our desired form that is at least as good. 


\vspace{3mm}

We first show that there is at most one $w_{\vec{q}_1} > 0$ where $\vec{q}_1 \in Q_1$ in any optimal solution. Corresponding to the vectors in $Q_1$, we define the following  function $f$. For all $1 - i^* \leq x \leq c - i^*$, let  
$$f(x) = -1/2 \cdot ((x + i^*)(x + i^* - 1) - i^*(i^* - 1)).$$
Note that for all $\vec{q} = (q_x, q_y) \in Q_1$, $f(q_x) = q_y$. Moreover, function $f$ is concave over its domain and strictly decreasing.

\begin{proposition}
    Let $\{w_{\vec{q}}\}_{\vec{q} \in Q_1 \cup Q_2}$ be a feasible solution to \text{\normalfont LP1}. Then we can obtain a feasible solution      $\{w'_{\vec{q}}\}_{\vec{q} \in Q_1 \cup Q_2}$
    that is at least as good as $\{w_{\vec{q}}\}$, by letting
    \begin{align*}
            w'_{(x, f(x))} &:= 0 & \text{\normalfont for } & \quad x \in [1 - i^*, -2] \\
            w'_{(-1, f(-1))} & :=  -\sum_{x  \in [1 - i^*, -1] } x \cdot w_{(x, f(x))} \\
w'_{(0, f(0))} & := 0 \\
        w'_{(1, f(1))} & := \sum_{x \in [1, c-i^*]} x \cdot w_{(x, f(x))} \\
         w'_{(x, f(x))} & := 0 & \text{\normalfont for } & \quad x \in [2, c - i^*] \\
        w'_{\vec{q}} & := w_{\vec{q}}   & \text{\normalfont for }& \quad \vec{q} \in Q_2
    \end{align*}
    \label{prop:2q1}
\end{proposition}
\begin{proof} \leavevmode
\begin{itemize}
\item \textbf{Feasible:} We must verify that $w_{\vec{q}} \geq 0$ for all $\vec{q} \in Q_1 \cup Q_2$ as required.  This follows straightforwardly from the definition of our solution and the fact that $w_{\vec{q}} \geq 0$ for all $ \vec{q} \in Q_1 \cup Q_2$, since we assumed our initial solution was feasible.

Additionally, the constraint $\sum_{\vec{q} \in Q_1 \cup Q_2} q_x w'_{\vec{q}} = v_x - u_x$ remains satisfied, since
\begin{align*}
\sum_{\vec{q} \in Q_1 \cup Q_2} q_x w'_{\vec{q}} & = 1 \cdot w'_{(1, f(1))} + -1 \cdot w'_{(-1, f(-1))} + \sum_{\vec{q} \in Q_2}  q_x w'_{\vec{q}} \\
& =  \sum_{x \in [1, c-i^*]} x \cdot w_{(x, f(x))} + \sum_{x  \in [1 - i^*, -1] } x \cdot w_{(x, f(x))}  + \sum_{\vec{q} \in Q_2} q_x w_{\vec{q}}  \\
& = \sum_{\vec{q} \in Q_1 \cup Q_2} q_x w_{\vec{q}} \\
& = v_x - u_x
\end{align*}
where the final equality follows from our assumption that our initial solution was feasible. 
\item \textbf{At least as good:}  
The value of our old solution is
\[
\sum_{\vec{q} \in Q_1 \cup Q_2} q_y  w_{\vec{q}} = \sum_{x \in [1 - i^*, c - i^*]}f(x) \cdot w_{(x, f(x))} + \sum_{\vec{q} \in Q_2} q_y w_{\vec{q}} 
\]
The value of our new solution is 
\begin{align*}
    \sum_{\vec{q} \in Q_1 \cup Q_2} q_y  w'_{\vec{q}} & =   \sum_{x \in [1 - i^*, c - i^*]}f(x) \cdot w'_{(x, f(x))} + \sum_{\vec{q} \in Q_2} q_y w_{\vec{q}} \\
    & = f(-1) \cdot w'_{(-1, f(-1))} + f(1) \cdot w'_{(1, f(1))} + \sum_{\vec{q} \in Q_2} q_y w_{\vec{q}} \\
    & = -f(-1) \cdot \sum_{x \in [1 - i^*, -1]} x \cdot w_{(x, f(x))} + f(1) \cdot   \sum_{x \in [1, c-i^*]} x \cdot w_{(x, f(x))} + \sum_{\vec{q} \in Q_2} q_y w_{\vec{q}} \\
     & = -(i^*-1) \cdot \sum_{x \in [1 - i^*, -1]} x \cdot w_{(x, f(x))} - i^* \cdot   \sum_{x \in [1, c-i^*]} x \cdot w_{(x, f(x))} + \sum_{\vec{q} \in Q_2} q_y w_{\vec{q}} 
\end{align*}

Then the difference between   the value of our new solution and our old solution is 
\begin{align*}
 \texttt{diff} := \sum_{x \in [1 - i^*, -1]} (( - i^* + 1)x - f(x) )  \cdot w_{(x, f(x))} +   \sum_{x \in [1, c-i^*]} (-i^*x - f(x)) \cdot w_{(x, f(x))}
\end{align*}
If we can establish that $f(x) \leq (-i^* + 1)x$ for $ x\leq -1$ and $f(x) \leq -i^*x$ for $x \geq 1$, then it will immediately follow that $\texttt{diff} \geq 0$. Observe that the line tangent to $f$ at $x=-1$ is $t(x) = (-i^* + 3/2)x + 1/2$. Then since $f$ is concave, $f(x) \leq t(x) \leq (-i^* + 1)x$ for $x \leq -1$. Likewise, the line tangent to $f$ at $x  = 1$ is $t(x) = (-i^* -1/2)x + 1/2$, so $f(x) \leq t(x) \leq -i^*x$ for $x \geq 1$.  We conclude that $\texttt{diff} \geq 0$, so our new solution is at least as good as our initial solution as desired. \qedhere
\end{itemize}
\end{proof}
We will use Proposition \ref{prop:2q1} to establish our desired property about weights $w_{\vec{q}_1}$ with $\vec{q}_1 \in Q_1$. 
\begin{proposition}
 There exist optimal solutions to \text{\normalfont LP1} with at most one $\vec{q}_1 \in Q_1$ such that $w_{\vec{q}_1} > 0$. 
\label{prop:1q1}
\end{proposition}
\begin{proof}
By Proposition \ref{prop:2q1}, we may assume our solution is of the form $\delta_1' = f(-1) \cdot w_{(-1, f(-1))} + f(1) \cdot w_{(1, f(1))} + \sum_{\vec{q} \in Q_2} q_y w_{\vec{q}}$ where $-1 \cdot w_{(-1, f(-1))} + 1 \cdot w_{(1, f(1))} + \sum_{\vec{q} \in Q_2} q_x w_{\vec{q}} = v_x - u_x$. We split our analysis into two cases:
\paragraph{Case 1: $w_{(-1, f(-1))} \leq w_{(1, f(1))}$} In this case, we define our new solution $\{w'_{\vec{q}}\}$ to be   
\begin{align*}
    w'_{(-1, f(-1))} &:= 0 \\
    w'_{(1, f(1))} &:= w_{(1, f(1))} - w_{(-1, f(-1))} \\
    w'_{\vec{q}} & := w_{\vec{q}} & \qquad \text{ for } \vec{q } \in Q_2
\end{align*}
It is immediate that  $\{w'_{\vec{q}}\}$ is feasible. Observe that the difference between the value of our new solution and our old solution is 
\begin{align*}
\texttt{diff} & = (f(1)(w_{(1, f(1))} - w_{(-1, f(-1))}))  - ( f(-1) \cdot w_{(-1, f(-1))} + f(1) \cdot w_{(1, f(1))}) \\
  & =  -f(1) \cdot w_{(-1, f(-1))} - f(-1) \cdot w_{(-1, f(-1))} \\
  & \geq (i^* - (i^* - 1)) \cdot   w_{(-1, f(-1))} \geq 0
\end{align*}
so our new solution is at least as good as the old solution.

\paragraph{Case 2:  $w_{(-1, f(-1))} > w_{(1, f(1))}$}  In this case, we define our new solution $\{w'_{\vec{q}}\}$ to be   
\begin{align*}
    w'_{(-1, f(-1))} &:=  w_{(-1, f(-1))} -  w_{(1, f(1))}  \\
    w'_{(1, f(1))} &:=0 \\
    w'_{\vec{q}} & := w_{\vec{q}} & \qquad \text{ for } \vec{q } \in Q_2
\end{align*}
It is immediate that  $\{w'_{\vec{q}}\}$ is feasible. Observe that the difference between the value of our new solution and our old solution is 
\begin{align*}
\texttt{diff} & = (f(-1)(w_{(-1, f(-1))} - w_{(1, f(1))}  ))  - ( f(-1) \cdot w_{(-1, f(-1))} + f(1) \cdot w_{(1, f(1))}) \\
  & =  -f(-1) \cdot w_{(1, f(1))} - f(1) \cdot w_{(1, f(1))} \\
  & \geq (-(i^*-1) + i^*) \cdot   w_{(1, f(1))} \geq 0
\end{align*}
Our new solution is at least as good as the old solution. The claim is established.
\end{proof}

Now to finish our proof of Proposition \ref{prop:optim}, we need to extend Proposition \ref{prop:1q1} to vectors in $Q_2$. 

\begin{proposition}
There exist optimal solutions $\{w_{\vec{q}}\}_{\vec{q} \in Q_1 \cup Q_2}$ to \text{\normalfont LP1} with at most one $\vec{q}_1 \in Q_1$ such that $w_{\vec{q}_1} > 0$ and at most one $\vec{q}_2 \in Q_2$ such that  $w_{\vec{q}_2} > 0$. 
\label{prop:lp12}
\end{proposition}
\begin{proof}
By Proposition \ref{prop:2q1}, we may assume our solution is of the form $\delta_1' = q_y^1 \cdot  w_{\vec{q}_1} + \sum_{\vec{q} \in Q_2} q_y \cdot w_{\vec{q}}$ where $\vec{q}_1 = (q_x^1, q_y^1) \in Q_1$, and $ q_x^1 \cdot  w_{\vec{q}_1} + \sum_{\vec{q} \in Q_2} q_x \cdot w_{\vec{q}} = v_x - u_x$. Now corresponding to the vectors in $Q_2$, we define the following  function $g$. For all $-2r_I - i^* + c  \leq x \leq -2r_I - i^* + 2c - 1$, let  
$$g(x) = 1/2 \cdot ((x + 2r_I - 2c + i^*)(x + 2r_I - 2c + i^* + 1) + i^*(i^* - 1)) -c(c+1) .$$
Note that for all $\vec{q} = (q_x, q_y) \in Q_2$, $g(q_x) = q_y$. Function $g$ is strictly decreasing. Moreover, the domain and range of $g$ contain only negative numbers. Let $\ell := -2r_I - i^* + c$.
We define the following new solution $\{w'_{\vec{q}}\}_{\vec{q} \in Q_1 \cup Q_2}$. 
\begin{align*}
    w'_{\vec{q}_1} &:= w_{\vec{q}_1} \\
    w'_{(\ell, g(\ell))} &:= \ell ^{-1} \cdot \sum_{x \in [\ell , \ell + c-1]} x \cdot w_{(x, g(x))} \\
    w'_{\vec{q}} & := 0 & \text{ if } \vec{q} \not \in \{ \vec{q}_1, (\ell, g(\ell)) \}
\end{align*}
    \begin{itemize}
        \item \textbf{Feasible:} Since $\ell^{-1} < 0$ and $x \leq \ell + c - 1 < 0$, it follows that $w'_{\vec{q}} \geq 0$ for all $\vec{q} \in Q_1 \cup Q_2$. Additionally, 
        $$q_x^1 \cdot  w_{\vec{q}_1} + \ell \cdot \ell^{-1} \cdot  \sum_{x \in [\ell , \ell + c-1]} x \cdot w_{(x, g(x))}   =  q_x^1 \cdot  w_{\vec{q}_1} + \sum_{\vec{q} \in Q_2} q_x \cdot w_{\vec{q}} = v_x - u_x$$
        so our constraint is satisfied.
        \item \textbf{Optimal:} The difference between our new solution and our old solution is 
        \begin{align*}
        \texttt{diff}  & =   \left( g(\ell) \cdot \ell ^{-1} \cdot \sum_{x \in [\ell , \ell + c-1]} x \cdot w_{(x, g(x))} \right) - \sum_{x \in [\ell, \ell + c - 1]} g(x) \cdot w_{(x, g(x))} \\
       & =  \sum_{x \in [\ell, \ell + c - 1]} (g(\ell) \cdot \ell^{-1} \cdot x - g(x)) \cdot w_{(x, g(x))} 
        \end{align*}
        If we can establish that $g(\ell) \cdot \ell^{-1} \cdot x - g(x) \geq 0$, then $\texttt{diff} \geq 0$. Since $\ell < 0$, this inequality can be restated as
        \[
        g(\ell) \cdot x \leq g(x) \cdot \ell 
        \]
        Now note that $\ell \leq x < 0$ and $g(x) \leq g(\ell) < 0$, since $g$ is strictly decreasing over its domain. Then $0 < -x \leq -\ell$ and $0 < -g(\ell) \leq -g(x)$, so combining these inequalities gives us $g(\ell) \cdot x \leq g(x) \cdot \ell$, so $\texttt{diff} \geq 0$, as desired.  \qedhere
    \end{itemize}
\end{proof}

Proposition \ref{prop:optim} follows immediately from Proposition \ref{prop:lp12}.


\end{appendix}

\end{document}